\documentclass[runningheads]{llncs}

\def\ourgrants{This work is part of a project that has received
  funding from the European Research Council (ERC) under the European
  Union's Horizon 2020 research and innovation programme (Grant
  agreement No.~852769, ARiAT). It is also supported by a Discovery
  Grant from the Natural Sciences and Engineering Research Council of
  Canada (NSERC). Parts of this research were carried out while the
  second author was affiliated with the Department of Computer
  Science, University College London, UK.}

\protected\def\ourthanks{\ourgrants\ \\ \\ }
\def\logoshift{44pt}

\usepackage{thm-restate}
\usepackage[hidelinks]{hyperref} 
\usepackage{tikzpagenodes}       
\newcommand{\logowidth}{20pt}    
\newcommand{\logos}{
  {\begin{tikzpicture}[remember picture, overlay]%
      \node[anchor=south east,inner sep=0pt,
            xshift=5+\logowidth, yshift=\logoshift]%
      at (current page text area.south east) {%
        \begin{tabular}{l}%
          \includegraphics[width=\logowidth]{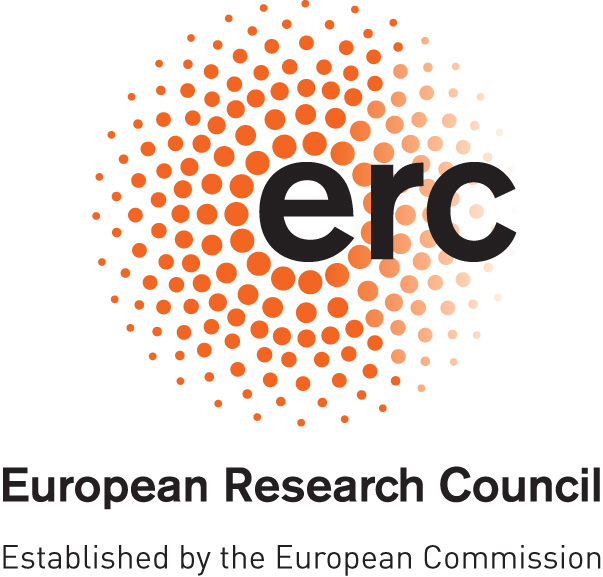} \\%
          \includegraphics[width=\logowidth]{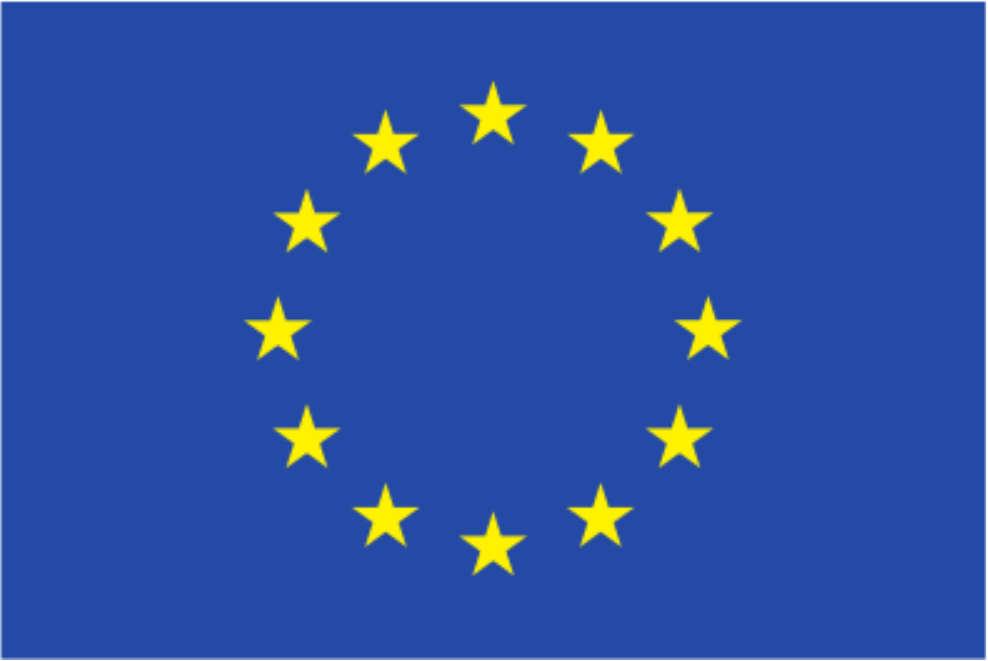}%
        \end{tabular}%
      };%
  \end{tikzpicture}}%
}

\usepackage{amsmath, amssymb}
\usepackage{xcolor}
\usepackage{cleveref}
\usepackage{listings}
\usepackage[noline, noend]{algorithm2e}
\let\G\relax            
\usepackage{complexity} 
\usepackage{macros}     
\usepackage{tikz}
\usepackage{pgfplots}
\pgfplotsset{compat=1.15}
\usepackage{mathtools}
\usepackage{tikz}
\usetikzlibrary{automata, positioning, calc, petri, hobby}
\usepackage{wrapfig}
\usepackage{booktabs}
\usepackage{multirow}

\begin{document}

\title{Directed Reachability for Infinite-State Systems\thanks{\ourthanks}}

\renewcommand{\orcidID}[1]{\href{https://orcid.org/#1}{\,\includegraphics[width=9pt]{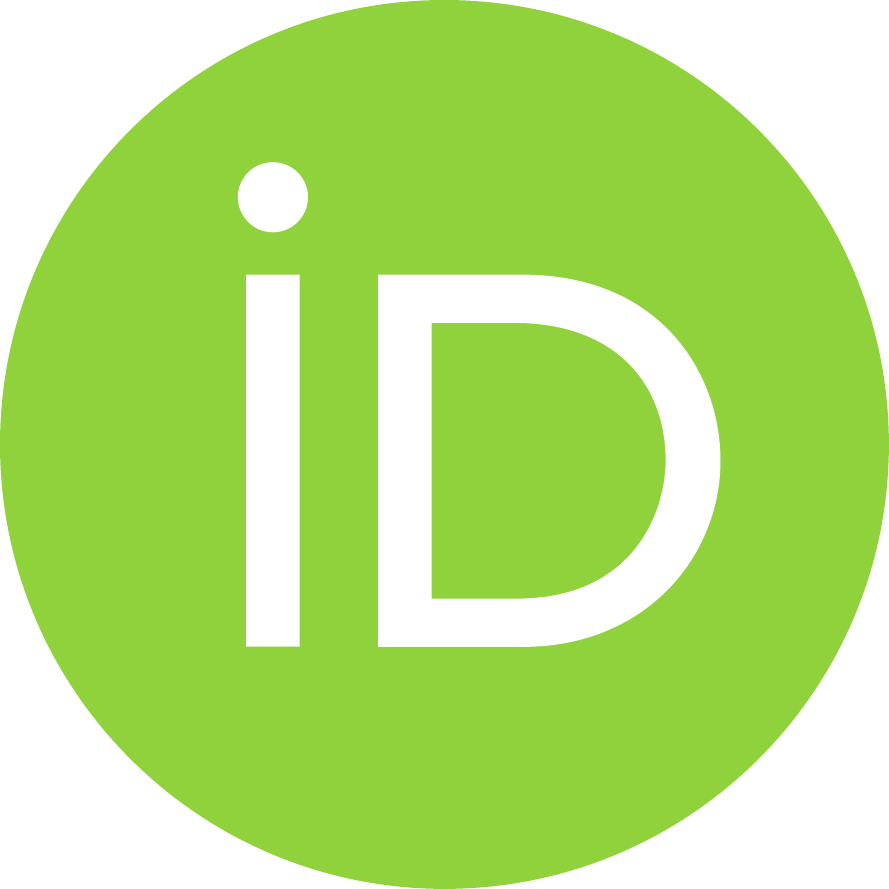}}} 

\author{Michael Blondin\inst{1}\orcidID{0000-0003-2914-2734} \and
Christoph Haase\inst{2}\orcidID{0000-0002-5452-936X} \and
Philip Offtermatt\inst{1,3}\orcidID{0000-0001-8477-2849}}

\authorrunning{M. Blondin et al.}

\institute{%
  Universit\'{e} de Sherbrooke, Canada \and  
  %
  University of Oxford, United Kingdom \and
  %
  Max Planck Institute for Software Systems, Saarbr\"{u}cken, Germany
}

\maketitle

\begin{abstract}
  Numerous tasks in program analysis and synthesis reduce to deciding
  reachability in possibly infinite graphs such as those induced by
  Petri nets. However, the Petri net reachability problem has recently
  been shown to require non-elementary time, which raises questions
  about the practical applicability of Petri nets as target models. In
  this paper, we introduce a novel approach for efficiently
  semi-deciding the reachability problem for Petri nets in
  practice. Our key insight is that computationally lightweight
  over-approximations of Petri nets can be used as distance oracles in
  classical graph exploration algorithms such as \astar\ and greedy
  best-first search. We provide and evaluate a prototype
  implementation of our approach that outperforms existing
  state-of-the-art tools, sometimes by orders of magnitude, and which
  is also competitive with domain-specific tools on benchmarks coming
  from program synthesis and concurrent program analysis.
\end{abstract}

\section{Introduction}
\label{sec:introduction}\logos%
Many problems in program analysis, synthesis and verification reduce
to deciding reachability of a vertex or a set of vertices in infinite
graphs, \eg, when reasoning about concurrent programs with an
unbounded number of threads, or when arbitrarily many components can
be used in a synthesis task. For automated reasoning tasks, those
infinite graphs are finitely represented by some mathematical model.
Finding the right such model requires a trade-off between the two
conflicting goals of maximal expressive power and computational
feasibility of the relevant decision problems. Petri nets are a
ubiquitous mathematical model that provides a good compromise between
those two goals. They are expressive enough to find a plethora of
applications in computer science, in particular in the analysis of
concurrent processes, yet the reachability problem for Petri nets is
decidable~\cite{May81,Kos82,Lam92,Ler12}. \emph{Counter abstraction}
has evolved as a generic abstraction paradigm that reduces a variety
of program analysis tasks to problems in Petri nets or variants
thereof such as well-structured transition systems, see
\eg~\cite{GS92,KKW14,ZP04,BKWZ13}. Due to their generality and
versatility, Petri nets and their extensions find numerous
applications also in other areas, including the design and analysis of
protocols~\cite{EGLM17}, business processes~\cite{Aal98}, biological
systems~\cite{HGD08,Cha07} and chemical systems~\cite{ADS07}. The goal
of this paper is to introduce and evaluate an efficient generic
approach to deciding the Petri net reachability problem on instances
arising from applications in program verification and synthesis.

A Petri net comprises a finite set of \emph{places} with a finite
number of \emph{transitions}. Places carry a finite yet unbounded
number of \emph{tokens} and transitions can remove and add tokens to
places. A \emph{marking} specifies how many tokens each place
carries. An example of a Petri net is given on the left-hand side of
\Cref{fig:ex:pn}, where the two places $\{p_1, p_2\}$ are depicted as
circles and transitions $\{t_1, t_2, t_3\}$ as squares. Places carry
tokens depicted as filled circles; thus $p_1$ carries one token and
$p_2$ carries none. We write this as $[p_1 \colon 1, p_2
  \colon 0]$, or $(1,0)$ if there is a clear ordering on the
places. Transition $t_1$ can add a single token to place $p_1$ at any
moment. As soon as a token is present in $p_1$, it can be consumed by
transition $t_2$, which then adds a token to place $p_2$ and puts back
one token to place $p_1$. Finally, transition $t_3$ consumes tokens
from $p_1$ without any adding token at all.
%
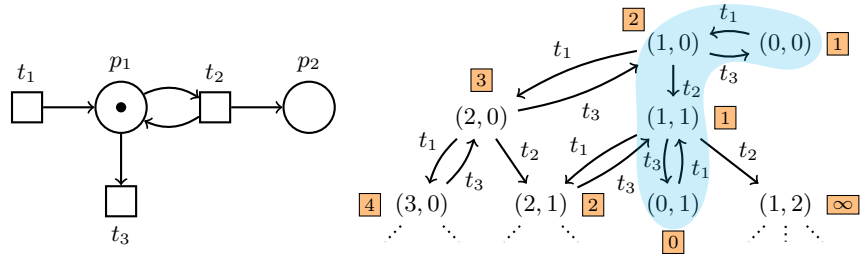
\begin{figure}[!h]
  \begin{minipage}[c]{.48\textwidth}
  \hspace*{5pt}
  \begin{tikzpicture}[auto, thick, node distance=1.25cm]
    \node[transition, label=above:{$t_1$}]              (t1) {};
    \node[place,      label=above:{$p_1$}, right of=t1, tokens=1] (p1) {};
    \node[transition, label=above:{$t_2$}, right of=p1] (t2) {};
    \node[place,      label=above:{$p_2$}, right of=t2] (p2) {};
    \node[transition, label=below:{$t_3$}, below of=p1] (t3) {};
    
    \path[->]
    (t1) edge[]             node[] {} (p1)
    (t2) edge[bend left=30] node[] {} (p1)
    (p1) edge[bend left=30] node[] {} (t2)
    (t2) edge[]             node[] {} (p2)
    (p1) edge[]             node[] {} (t3)
    ;
  \end{tikzpicture}%
  \vspace*{-10pt} 
\end{minipage}%
\hspace*{-45pt}
\begin{minipage}[c]{.48\textwidth}
  \vspace*{-8pt} 
  \begin{tikzpicture}[auto, thick, node distance=1cm
                      transform shape, scale=0.9]
    \newcommand{\mshift}{12pt}
    
    \node[]                                     (m0) {$(0, 0)$};
    \node[left       =              15pt of m0] (m1) {$(1, 0)$};
    \node[below left =\mshift   and 45pt of m1] (m2) {$(2, 0)$};
    \node[below      =\mshift   and      of m1] (m3) {$(1, 1)$};
    \node[below left =\mshift+5 and -5pt of m2] (m4) {$(3, 0)$};
    \node[below right=\mshift+5 and -5pt of m2] (m5) {$(2, 1)$};
    \node[below      =\mshift+5 and      of m3] (m6) {$(0, 1)$};
    \node[below right=\mshift+5 and 15pt of m3] (m7) {$(1, 2)$};

    \tikzstyle{hscore} = [draw, font=\scriptsize, thin, inner sep=2pt,
                          fill=orange!50];
    \newcommand{\hpos}{1.75pt}
    
    \node[hscore, right=\hpos of m0] {$1$};
    \node[hscore, above=\hpos of m1, xshift=-15pt, yshift=-5pt] {$2$};
    \node[hscore, above=\hpos of m2] {$3$};
    \node[hscore, right=\hpos of m3] {$1$};
    \node[hscore,  left=\hpos of m4] {$4$};
    \node[hscore, right=\hpos of m5] {$2$};
    \node[hscore, below=\hpos of m6] {$0$};
    \node[hscore, right=\hpos of m7] {$\infty$};
    
    \newcommand{\labshift}{-3pt}

    \path[->]
    (m0) edge[bend right=15] node[swap]                    {$t_1$} (m1)
    (m1) edge[bend right=15] node[swap]                    {$t_3$} (m0)
    (m1) edge[bend right=10] node[xshift=-\labshift, swap] {$t_1$} (m2)
    (m2) edge[bend right=10] node[xshift= \labshift, swap] {$t_3$} (m1)
    (m1) edge[]              node[yshift= \labshift]       {$t_2$} (m3)
    (m2) edge[bend right=15] node[xshift=-\labshift, swap] {$t_1$} (m4)
    (m4) edge[bend right=15] node[xshift= \labshift, swap] {$t_3$} (m2)
    (m2) edge[]              node[yshift= \labshift]       {$t_2$} (m5)
    (m3) edge[bend right=08] node[yshift= \labshift, swap] {$t_1$} (m5)
    (m5) edge[bend right=08] node[xshift= \labshift, swap] {$t_3$} (m3)
    (m3) edge[bend right=15] node[xshift=-\labshift, swap] {$t_3$} (m6)
    (m6) edge[bend right=15] node[yshift= \labshift, swap] {$t_1$} (m3)
    (m3) edge[]              node[yshift= \labshift]       {$t_2$} (m7)
    ;
    
    \newcommand{\mghost}{\phantom{($0, 0, 0)$}}
    \node[below  left=\mshift/2 and -10pt of m4] (d1) {\mghost};
    \node[below right=\mshift/2 and -10pt of m4] (d2) {\mghost};
    \node[below  left=\mshift/2 and -10pt of m5] (d3) {\mghost};
    \node[below right=\mshift/2 and -10pt of m5] (d4) {\mghost};
    \node[below  left=\mshift/2 and -10pt of m7] (d5) {\mghost};
    \node[below      =\mshift/2 and       of m7] (d6) {\mghost};
    \node[below right=\mshift/2 and -10pt of m7] (d7) {\mghost};
  
    \path[-, dotted]
    (m4) edge node {} (d1)
    (m4) edge node {} (d2)
    (m5) edge node {} (d3)
    (m5) edge node {} (d4)
    (m7) edge node {} (d5)
    (m7) edge node {} (d6)
    (m7) edge node {} (d7)
    ;

    \fill[cyan, opacity=0.2] ([shift={(0pt, 4.5pt)}]m0.north)
    to[closed, curve through={
        ([shift={ (0pt,  2.5pt)}]m1.north)
        ([shift={ (3pt,    0pt)}]m1.west)
        ([shift={ (0pt,    0pt)}]m3.west)
        ([shift={ (3pt,    0pt)}]m6.west)
        ([shift={ (0pt, -1.5pt)}]m6.south)
        ([shift={(-3pt,    0pt)}]m6.east)
        ([shift={(-2pt,    0pt)}]m3.east)
        ([shift={ (2pt,   -2pt)}]m1.south east)
        ([shift={ (0pt,   -2pt)}]m0.south)
        ([shift={ (0pt,    0pt)}]m0.east)
      }
    ]
    ([shift={(0pt, 4.5pt)}]m0.north);
  \end{tikzpicture}%
  \vspace*{-10pt} 
\end{minipage}%
  \caption{\emph{Left:} A Petri net $\pn$. \emph{Right:} Search of the
    forthcoming Algorithm~\ref{alg:directed-search} over the graph
    $G_\N(\pn)$ from $(0, 0)$ to $(0, 1)$, where $(x, y)$ denotes
    $[p_1\colon x, p_2\colon y]$ and each number in a box next to a
    marking is its heuristic value. Only the blue region is
    expanded.}\label{fig:search:ex}\label{fig:ex:pn}
\end{figure}

A Petri net induces a possibly infinite directed graph whose vertices
are markings, and whose edges are determined by the transitions of the
Petri net, \emph{cf.}\ the right side of \Cref{fig:ex:pn}. Given two
markings, the \emph{reachability problem} asks whether they are
connected in this graph. In \Cref{fig:ex:pn}, the marking $(0,1)$ is
reachable from $(0,0)$, \eg, via paths of lengths $3$ and~$5$: $(0,0)
\xrightarrow{t_1} (1,0) \xrightarrow{t_2} (1,1) \xrightarrow{t_3}
(0,1)$ and $(0,0) \xrightarrow{t_1} (1,0) \xrightarrow{t_1} (2,0)
\xrightarrow{t_2} (2,1) \xrightarrow{t_3} (1,1) \xrightarrow{t_3}
(0,1)$.

In practice, the Petri net reachability problem is a challenging
decision problem due to its horrendous worst-case complexity: an
exponential-space lower bound was established in the
1970s~\cite{Lip76}, and a non-elementary time lower bound has only
recently been established~\cite{CLLLM19}. One may thus question
whether a problem with such high worst-case complexity is of any
practical relevance, and whether reducing program analysis tasks to
Petri net reachability is anything else than merely an intellectual
exercise. We debunk those concerns and present a technique which
decides most reachability instances appearing in the wild. When
evaluated on large-scale instances involving Petri nets with thousands
of places and tens of thousands of transitions, our prototype
implementation is most of the time faster, even up to several orders
of magnitude on large-scale instances, and solves more instances than
existing state-of-the-art tools. Our implementation is also
competitive with specialized domain-specific tools. One of the biggest
advantages of our approach is that it is extremely simple to describe
and implement, and it readily generalizes to many extensions of Petri
nets. In fact, it was surprising to us that our approach has not yet
been discovered. We now describe the main observations and techniques
underlying our approach.

Ever since the early days of research in Petri nets, state-space
over-approxi\-ma\-tions have been studied to attenuate the high
computational complexity of their decision problems. One such
over-approxi\-ma\-tion is, informally speaking, to allow places to
carry a negative numbers of tokens. Deciding reachability then reduces
to solving the so-called \emph{state equation}, a system of linear
equations associated to a Petri net. Another over-approximation are
\emph{continuous Petri nets}, a variant where places carry fractional
tokens and ``fractions of transitions'' can be
applied~\cite{DA87}. The benefit is that deciding reachability drops
down to polynomial time~\cite{FH15}. While those approximations have
been applied for pruning search spaces, see
\eg~\cite{ELMMN14,ALW16,BFHH17,GLS18}, we make the following simple
key observation:
\begin{quotation}
  \emph{If a marking $\vec{m}$ is reachable from an initial marking in
    an over-approxi\-ma\-tion, then the length of a shortest
    witnessing path in the over-approxi\-ma\-tion lower bounds the
    length of a shortest path reaching $\vec{m}$.}
\end{quotation}
The availability of an oracle providing lower bounds on the length of
shortest paths between markings enables us to appeal to classical
graph traversal algorithms which have been highly successful in
artificial intelligence and require such oracles, namely \astar\ and
greedy best-first search, see \eg~\cite{RN09}. In particular,
determining the length of shortest paths in the over-approximations
described above can be phrased as optimization problems in (integer)
linear programming and optimization modulo theories, for which
efficient off-the-shelf solvers are available~\cite{Gur20,BPF15}.
Thus, oracle calls can be made at comparably modest computational
cost, which is crucial for the applicability of those algorithms. As a
result, a large class of existing state-space over-approximations can
be applied to obtain a highly efficient forward-analysis semi-decision
procedure for the reachability problem. For example, in
\Cref{fig:search:ex}, using the state equation as distance oracle,
\astar\ only explores the four vertices in the blue region and
directly reaches the target vertex, whereas a breadth-first search may
need to explore all vertices of the figure and a depth-first search
may even not terminate.

In theory, our approach could be turned into a decision procedure by
applying bounds on the length of shortest paths in Petri
nets~\cite{LS15}. However, such lengths can grow non-elementarily in
the number of places~\cite{CLLLM19}, and just computing the cut-off
length will already be infeasible for any Petri net of practical
relevance. It is worth mentioning that, in practice, it has been
observed that the over-approximations we employ also often witness
non-reachability though, see \eg~\cite{ELMMN14}. Still, when dealing
with finite state spaces, our procedure is complete.

A noteworthy benefit of our approach is that it enables finding
\emph{shortest} paths when \astar\ is used as the underlying
algorithm. In program analysis, paths usually correspond to traces
reaching an erroneous configuration. In this setting, shorter error
traces are preferred as they help understanding why a certain error
occurs. Furthermore, in program synthesis, paths correspond to
synthesis plans. Again, shorter paths are preferred as they yield
shorter synthesized programs. In fact, we develop our algorithmic
framework for \emph{weighted} Petri nets in which transitions are
weighted with positive integers. Classical Petri nets correspond to
the special instance where all weights are equal to one. Weighted Petri nets
are useful to reflect cost or preferences in synthesis tasks. For
example, there are program synthesis approaches where software
projects are mined to determine how often API methods are called to
guide a procedure by preferring more frequent
methods~\cite{GRBHS14,Gal14,LDZ19}. Similarity metrics can also be
used to obtain costs estimating the relevance of invoking
methods~\cite{FMWDR17}. It has further been argued that weighted Petri
nets are a good model for synthesis tasks of chemical reactions as
they can reflect costs of various chemical
compounds~\cite{WWB17}. Finally, weights can be viewed as representing
an amount of time it takes to fire a transition, see \eg~\cite{Mur89}.

\paragraph*{Related work.}

Our approach falls under the umbrella term \emph{directed model
  checking} coined in the early 2000s, which refers to a set of
techniques to tackle the state-explosion problem via guided
state-space exploration. It primarily targets disproving safety
properties by quickly finding a path to an error state without the
need to explicitly construct the whole state space. As such, directed
model checking is useful for bug-finding since, in the words of Yang
and Dill~\cite{YD98}, \emph{in practice, model checkers are most
  useful when they find bugs, not when they prove a property}. The
survey paper~\cite{ESBWFA09} gives an overview over various directed
model checking techniques for finite-state systems.

For Petri nets, directed reachability algorithms based on
over-approximations as developed in this work have not been
described. In~\cite{UP93}, it is argued that exploration heuristics,
like \astar, can be useful for Petri nets, but they do not consider
over-approximations for the underlying heuristic functions. The
authors of~\cite{JC98} use Petri nets for scheduling problems and
employ the state equation, viewed as a system of linear equations over
$\Q$, in order to explore and prune reachability graphs. This approach
is, however, not guaranteed to discover shortest paths. There has been
further work on using \astar\ for exploring the reachability graph of
Petri nets for scheduling problems, see, \eg, \cite{LD94,MO05} and the
references therein.




\section{Preliminaries}
\label{sec:preliminaries}
Let $\N \defeq \{0, 1, \ldots\}$. For all $\D \subseteq \Q$ and
${\succ} \in \{\geq, >\}$, let $\D_{\succ 0} \defeq \{a \in \D : a
\succ 0\}$, and for every set $X$, let $\D^X$ denote the set of
vectors $\D^X \defeq \{\vec{v} \mid \vec{v} \colon X \to \D\}$. We
naturally extend operations componentwise. In particular, $(\vec{u} +
\vec{v})(x) \defeq \vec{u}(x) + \vec{v}(x)$ for every $x \in X$, and
$\vec{u} \geq \vec{v}$ iff $\vec{u}(x) \geq \vec{v}(x)$ for every $x
\in X$.

\paragraph*{Graphs.}

A \emph{(labeled directed) graph} is a triple $G = (V, E, A)$, where
$V$ is a set of \emph{nodes}, $A$ is a finite set of elements called
\emph{actions}, and $E \subseteq V \times A \times V$ is the set of
\emph{edges} labeled by actions. We say that $G$ has \emph{finite
  out-degree} if the set of outgoing edges $\{(w, a, w') \in E : w =
v\}$ is finite for every $v \in V$. Similarly, it has \emph{finite
  in-degree} if the set of ingoing edges is finite for every $v \in
V$. If $G$ has both finite out- and in-degree, then we say that $G$ is
\emph{locally finite}. A \emph{path} $\path$ is a finite sequence of
nodes ${(v_i)}_{1 \le i \le n}$ and actions ${(a_i)}_{1 \leq i < n}$
such that $(v_i, a_i, v_{i+1}) \in E$ for all $1 \le i < n$. We say
that $\path$ is a \emph{path from $v$ to $w$} (or a \emph{$v$-$w$
  path}) if $v = v_1$ and $w = v_n$, and its \emph{label} is $a_1 a_2
\cdots a_{n-1}$, where $\varepsilon$ denotes the empty sequence.

A \emph{weighted} graph is a tuple $G = (V, E, A, \mu)$ where $(V, E,
A)$ is a graph with a \emph{weight function} $\mu \colon E \to
\Qpos$. The \emph{weight} of path $\path$ is the weight of its edges,
\ie\ $\mu(\path) \defeq \sum_{1 \le i < n} \mu(v_i, a_i, v_{i+1})$. A
\emph{shortest path from $v$ to $w$} is a $v$-$w$ path $\path$
minimizing $\mu(\path)$. We define $\dist_G \colon V \times V \to
\Qnon \cup \{\infty\}$ as the \emph{distance function} where
$\dist_G(v, w)$ is the weight of a shortest path from $v$ to $w$, with
$\dist_G(v, w) \defeq \infty$ if there is none.
We assume throughout the paper that weighted graphs have a
\emph{minimal weight}, \ie\ that $\min\{\mu(e) : e \in E\}$
exists. For graphs with finite out-degree, this ensures that if a path
exists between two nodes, then a shortest one
exists.\footnote{Otherwise, there could be increasingly better paths,
  \eg\ of weights $1, 1/2, 1/4, \ldots$.} This mild assumption always
holds in our setting.

\paragraph*{Petri nets.}

A \emph{weighted Petri net} is a tuple $\pn = (P, T, f, \lambda)$
where
\begin{itemize}
\item $P$ is a finite set whose elements are called \emph{places},
  
\item $T$ is a finite set, disjoint from $P$, whose elements are
  called \emph{transitions},

\item $f \colon (P \times T) \cup (T \times P) \to \N$ is the
  \emph{flow function} assigning multiplicities to arcs connecting
  places and transitions, and

\item $\lambda \colon T \to \Qpos$ is the \emph{weight function}
  assigning weights to transitions.
\end{itemize}
A \emph{marking} is a vector $\vec{m} \in \N^P$ which indicates that
place $p$ holds $\vec{m}(p)$ \emph{tokens}. A weighted Petri net with
$\lambda(t) = 1$ for each $t \in T$ is called a \emph{Petri net}. For
example, \Cref{fig:ex:pn} depicts a Petri net $\pn$ with $P = \{p_1,
p_2\}$, $T = \{t_1, t_2, t_3\}$, $f(p_1, t_3) = \allowbreak f(p_1,
t_2) = \allowbreak f(t_1, p_1) = f(t_2, p_1) = f(t_2, p_2) = 1$
(multiplicity omitted on arcs) and $f(-, -) = 0$ elsewhere (no
arc). Moreover, $\pn$ is marked with $[p_1 \colon 1, p_2 \colon 0]$.

The \emph{guard} and \emph{effect} of a transition $t \in T$ are
vectors $\guard{t} \in \N^p$ and $\effect{t} \in \Z^p$ where
$\guard{t}(p) \defeq f(p, t)$ and $\effect{t}(p) \defeq f(t, p) - f(p,
t)$. We say that $t$ is \emph{firable} from marking $\vec{m}$ if
$\vec{m} \geq \guard{t}$. If $t$ is firable from $\vec{m}$, then it
may be \emph{fired}, which leads to marking $\vec{m}' \defeq \vec{m} +
\effect{t}$. We write this as $\vec{m} \reach{t} \vec{m}'$. These
notions naturally extend to sequences of transitions,
\ie\ $\reach{\varepsilon}$ denotes the identity relation over $\N^P$,
$\effect{\varepsilon} \defeq \vec{0}$, $\lambda(\varepsilon) \defeq
0$, and for every $t_1, t_2, \ldots, t_k \in T$: $\effect{t_1 t_2
  \cdots t_k} \defeq \effect{t_1} + \effect{t_2} + \cdots +
\effect{t_k}$, $\lambda(t_1 t_2 \cdots t_k) \defeq \lambda(t_1) +
\lambda(t_2) + \cdots + \lambda(t_k)$, and
\begin{align*}
  \reach{t_1 t_2 \cdots t_k}
  &\defeq\ \reach{t_k} \circ \cdots \circ \reach{t_2} \circ \reach{t_1},
\end{align*}
We say that $\reach{} \defeq \cup_{t \in T} \reach{t}{}$ and
$\reach{*} \defeq \cup_{\sigma \in T^*} \reach{\sigma}{}$ are the
\emph{step} and \emph{reachability} relations. Note that the latter is
the reflexive transitive closure of $\reach{}$.

For example, $\vec{m} \reach{t_2 t_3} \vec{m}'$ and $\vec{m}
\reach{t_1 t_2 t_3 t_3} \vec{m}'$ in \Cref{fig:ex:pn}, where $\vec{m}
\defeq [p_1 \colon 1, p_2 \colon 0]$ and $\vec{m}' \defeq [p_1 \colon
  0, p_2 \colon 1]$. Moreover, $t_2$ is not firable in $\vec{m}'$.

Given a sequence $\sigma \in T^*$, denote by $|\sigma|_t \in \N$ the
number of times transition $t$ occurs in $T$. The \emph{Parikh image}
of $\sigma$ is the vector $\vec{\sigma} \in \N^T$ that captures the
number of occurrences of transitions appearing in $\sigma$,
\ie\ $\vec{\sigma}(t) \defeq |\sigma|_t$ for all $t\in T$.

Each weighted Petri net $\pn = (P, T, f, \lambda)$ induces a locally
finite weighted graph $G_\N(\pn) \defeq (V, E, T, \mu)$, called its
\emph{reachability graph}, where $V \defeq \N^P$, $E \defeq
\{(\vec{m}, t, \vec{m}') : \vec{m} \reach{t} \vec{m}'\}$ and
$\mu(\vec{m}, t, \vec{m}') \defeq \lambda(t)$ for each $(\vec{m}, t,
\vec{m}') \in E$. An example of a reachability graph is given on the
right of \Cref{fig:ex:pn}. We write $\dist_\pn$ to denote
$\dist_{G_\N(\pn)}$. We have $\dist_\pn(\vec{m}, \vec{m}') \neq
\infty$ iff $\vec{m} \reach{\sigma} \vec{m}'$ for some $\sigma \in
T^*$, and if the latter holds, then $\dist_\pn(\vec{m}, \vec{m}')$ is
the minimal weight among such firing sequences $\sigma$. Moreover, for
(unweighted) Petri nets, $\dist_\pn(\vec{m}, \vec{m}')$ is the minimal
number of transitions to fire to reach $\vec{m}'$ from $\vec{m}$.


\section{Directed search algorithms}
\label{sec:search}
Our approach relies on classical pathfinding procedures guided by node
selection strategies. Their generic scheme is described in
Algorithm~\ref{alg:directed-search}. Its termination with a value
$d\neq \infty$ indicates that the weighted graph $G = (V, \allowbreak E,
\allowbreak A, \mu)$ has a path from $s$ to $t$ of weight $d$, whereas
termination with $d=\infty$ signals that $\dist_G(s, t) = \infty$.
\begin{wrapfigure}[14]{l}{0.6\textwidth}
\vspace*{-16pt}
\begin{algorithm}[H]
  \DontPrintSemicolon
  \LinesNumbered
  \everypar={\nl}
  
  $g := [s \mapsto 0, v \mapsto \infty : v \neq s]$\;
  $C := \{ s \}$\;
  
  \While{$C \neq \emptyset$}{
    $v := \argmin_{v \in C} S(g,v)$\;\label{ln:strat}
    \lIf{$v = t$\label{ln:target:if}}{
      \Return $g(t)$\label{ln:target:ret}
    }
    \For{$(v, a, w) \in E$\label{ln:successors}}{
      \If{$g(v) + \mu(v,a,w) < g(w)$}{
        $\mathmakebox[20pt][l]{g(w)} := g(v) + \mu(v,a,w)$\;\label{ln:update:g}
        $\mathmakebox[20pt][l]{C}    := C \cup \{ w \}$\;
      }
    }
    
    $C := C \setminus \{ v \}$\;
  }
  
  \Return $\infty$\;
  
  \caption{Directed search algorithm.}\label{alg:directed-search}
\end{algorithm}
\end{wrapfigure}
Algorithm~\ref{alg:directed-search} maintains a set of \emph{frontier
  nodes} $C$ and a mapping $g \colon V \to \Qnon \cup \{\infty\}$ such
that $g(w)$ is the weight of the best known path from $s$ to $w$. In
Line~\ref{ln:strat}, a \emph{selection strategy} $S$ determines which
node $v$ to \emph{expand} next. Starting from
Line~\ref{ln:successors}, a successor $w$ of $v$ is added to the
frontier if its distance improves.

Let $h \colon V \to \Qnon \cup \{\infty\}$ estimate the distance from
all nodes to a target $t \in V$. The selection strategies sending $(g,
v)$ respectively to $g(v)$, $g(v) + h(v)$ or $h(v)$ yield the
classical Dijkstra's, \astar{} and greedy best-first search
(\emph{GBFS}) algorithms.

When instantiating $S$ with Dijkstra's selection strategy, a return
value $d\neq \infty$ is guaranteed to equal $\dist_G(s, t)$. This is
not true for \astar\ and GBFS\@. However, if $h$ fulfills the following
\emph{consistency} properties, then \astar\ also has this guarantee:
$h(t) = 0$ and $h(v) \le \mu(v, a, w) + h(w)$ for every $(v, a, w) \in
E$ (see, \eg, \cite{RN09}).

In the setting of infinite graphs, unlike GBFS, \astar\ and Dijkstra's selection
strategies guarantee termination if $\dist_G(s, t) \neq \infty$.
Yet, we introduce \emph{unbounded heuristics} for which termination is
also guaranteed for GBFS\@.
Note that these guarantees would vanish in the presence of
zero weights. An \emph{infinite path} $\pi$ is a sequence of nodes
${(v_i)}_{i \in \N}$ and actions ${(a_i)}_{i \in \N}$ such that $(v_i,
a_i, v_{i+1}) \in E$ for all $i \in \N$. We say that $\pi$ is
\emph{bounded} w.r.t.\ $h$ if its nodes are pairwise distinct and
there exists $b \in \Qnon$ with $h(v_i) \leq b$ for all $i \geq 0$. We
say that $h$ is \emph{unbounded} if it admits no bounded sequence.
The following technical lemma enables to prove termination of GFBS in
the presence of unbounded heuristics.

\begin{restatable}{lemma}{LemGreedy}\label{lem:greedy}
  If $G$ is locally finite and $h$ is unbounded, then the following
  holds:
  \begin{enumerate}
  \item The set of paths of weight at most $c \in \Qnon$ starting from
    node $s$ is finite.\label{itm:path:fin}

  \item Let $W \subseteq V$. The set $\dist_G(W, t) \defeq
    \{\dist_G(w, t) : w \in W\}$ has a minimum.\label{itm:dist:min}

  \item No node is expanded infinitely often by
    Algorithm~\ref{alg:directed-search}.\label{itm:fin:exp}
  \end{enumerate}
\end{restatable}

\begin{theorem}\label{thm:greedy:termination}
  Algorithm~\ref{alg:directed-search} with the greedy best-first
  search selection strategy always finds reachable targets for
  locally finite graphs and unbounded heuristics.
\end{theorem}
\begin{proof}
   First observe that Algorithm~\ref{alg:directed-search} satisfies
   this invariant:
   \begin{multline}
     \text{if $g(v) \neq \infty$, then $g(v)$ is the weight of a path
       from $s$ to $v$ in $G$} \\[-3pt] \text{whose nodes were all
       expanded, except possibly $v$}.\tag*{($\ast$)}\label{eq:g:path}
   \end{multline}
  
  Assume $\dist_G(s, t) \neq \infty$. For the sake of contradiction,
  suppose $t$ is never expanded. Let $K_i$ be the subgraph of $G$
  induced by nodes expanded at least once within the first $i$
  iterations of the \textbf{while} loop. In particular, $K_1$ is the
  graph made only of node $s$. Let $K = K_1 \cup K_2 \cup \cdots$. By
  \Cref{lem:greedy}~\eqref{itm:fin:exp}, no node is expanded
  infinitely often, hence $K$ is infinite. Moreover, $K$ has finite
  out-degree, and each node of $K$ is reachable from $s$ in $K$
  by~\ref{eq:g:path}. Thus, by K\"{o}nig's lemma, $K$ contains an
  infinite path $v_0, v_1, \ldots \in V$ of pairwise distinct nodes.

  Let $w$ be a node of $K$ minimizing $\dist_G(w, t)$. It is
  well-defined by \Cref{lem:greedy}~\eqref{itm:dist:min}. We have
  $\dist_G(w, t) \neq \infty$ as $t$ is reachable from $s$ and the
  latter belongs to $K_1 \subseteq K$. By minimality of $w \neq t$,
  there exists an edge $(w, a, w')$ of $G$ such that $\dist_G(w', t) <
  \dist_G(w, t)$ and $w'$ does not appear in $K$. Note that $w'$ is
  added to $C$ at some point, but is never expanded as it would
  otherwise belong to $K$. Let $i$ be the smallest index such that $w$
  belongs to $K_i$. Since $h$ is unbounded, there exists $j$ such that
  $h(v_j) > h(w')$ and $v_j$ is expanded after iteration $i$ of the
  while loop. This is a contradiction as $w'$ would have been expanded
  instead of $v_j$.\qed
\end{proof}


\section{Directed reachability}
\label{sec:safety}
In this section, we explain how to instantiate
Algorithm~\ref{alg:directed-search} for finding short(est) firing
sequences witnessing reachability in weighted Petri nets. Since
Dijkstra's selection strategy does not require any heuristic, we focus
on \astar\ and greedy best-first search which require consistent and
unbounded heuristics. More precisely, we introduce distance
under-approxi\-ma\-tions (\Cref{ssec:relaxations}); present relevant
concrete distance under-approxi\-ma\-tions (\Cref{ssec:pn:relax}); and
put everything together into our framework
(\Cref{ssec:directed:safety}).

\subsection{Distance under-approximations}\label{ssec:relaxations}

A \emph{distance under-approximation} of a weighted Petri net $\pn
= (P, T, f, \lambda)$ is a function
$d \colon \N^P \times \N^P \to \Qnon
\cup \{\infty\}$ such that for all $\vec{m}, \vec{m}', \vec{m}'' \in \N^P$:
\begin{itemize}
\item $d(\vec{m}, \vec{m}') \leq \dist_\pn(\vec{m}, \vec{m}')$,

\item $d(\vec{m}, \vec{m}'') \leq d(\vec{m}, \vec{m}') + d(\vec{m}',
  \vec{m}'')$ (\emph{triangle inequality}), and

\item $d$ is effective, \ie\ there is an algorithm that evaluates $d$
  on all inputs.
\end{itemize}

We naturally obtain a heuristic from $d$ for a directed search towards
marking $\vec{m}_\text{target}$. Indeed, let
$h \colon \N^P \to \Qnon \cup \{\infty\}$ be defined by $h(\vec{m})
\defeq d(\vec{m}, \vec{m}_\text{target})$. The following proposition
shows that $h$ is a suitable heuristic for \astar:
\begin{proposition}\label{prop:h:consistent}
  Mapping $h$ is a consistent heuristic.
\end{proposition}

\begin{proof}
  Let $\vec{m}, \vec{m}' \in \N^P$ and $t \in T$ be such that $\vec{m}
  \reach{t} \vec{m}'$. We have:
  \begin{align*}
    h(\vec{m})
    &= d(\vec{m}, \vec{m}_\text{target})
    && \text{(by def.\ of $h$)} \\
    &\leq d(\vec{m}, \vec{m}') + d(\vec{m}', \vec{m}_\text{target})
    && \text{(by the triangle inequality)} \\
    &\leq \dist_\pn(\vec{m}, \vec{m}') + d(\vec{m}', \vec{m}_\text{target})
    && \text{(by distance under-approximation)} \\[-4pt]
    &\leq \lambda(t) + d(\vec{m}', \vec{m}_\text{target})
    && \text{(since $\vec{m} \reach{t} \vec{m}'$)} \\
    &= \lambda(t) + h(\vec{m}')
    && \text{(by def.\ of $h$)}.
  \end{align*}
  Moreover, $h(\vec{m}_\text{target}) = d(\vec{m}_\text{target},
  \vec{m}_\text{target}) \leq \dist_\pn(\vec{m}_\text{target},
  \vec{m}_\text{target}) = 0$, where the last equality follows from the fact
  that weights are positive.\qed
\end{proof}

\subsection{From Petri net relaxations to distance under-approximations}\label{ssec:pn:relax}

We now introduce classical relaxations of Petri nets which
over-approximate reachability and consequently give rise to distance
under-approximations. The main source of hardness of the reachability
problem stems from the fact that places are required to hold a
non-negative number of tokens. If we relax this requirement and allow negative
numbers of tokens, we obtain a more tractable relation. More precisely, we write
$\vec{m} \zreach{t} \vec{m}'$ iff $\vec{m}' = \vec{m}
+ \effect{t}$. Note that transitions are always firable under this
semantics. Moreover, they may lead to ``markings'' with negative
components.

Another source of hardness comes from the fact that markings are
discrete. Hence, we can further relax ${\zreach{}}$ into ${\qreach{}}$
where transitions may be scaled down:
\begin{align*}
  \vec{m} \qreach{t} \vec{m}'
  &\iff \vec{m}' = \vec{m} + \delta \cdot \effect{t}
  \text{ for some $0 < \delta \le 1$}.
\end{align*}
One gets a less crude relaxation from considering \emph{nonnegative}
``markings'' only:
\begin{align*}
  \vec{m} \creach{t} \vec{m}'
  &\iff (\vec{m} \geq \delta \cdot \guard{t}) \text{ and }
        (\vec{m}' = \vec{m} + \delta \cdot \effect{t})
  \text{ for some $0 < \delta \le 1$}.
\end{align*}

Under these, we obtain ``markings'' from $\Q^P$ and $\Qnon^P$
respectively. Petri nets equipped with relation $\creach{}$ are known
as \emph{continuous Petri nets}~\cite{DA87,DA10}.

To unify all three relaxations, we sometimes write $\vec{m}
\areach{\G}{\delta t} \vec{m}'$ to emphasize the scaling factor
$\delta$, where $\delta = 1$ whenever $\G = \Z$.
Let $d_\G \colon \N^P \times \N^P \to \Qnon \cup \{\infty\}$ be
defined as $d_\G(\vec{m}, \vec{m}') \defeq \infty$ if $\vec{m}
\not\greach{*} \vec{m}'$, and otherwise:
\[
  d_\G(\vec{m}, \vec{m}') \defeq \min\left\{\sum_{i=1}^n \delta_i
  \cdot \lambda(t_i) :  \vec{m}
    \greach{\delta_1 t_1 \cdots \delta_n t_n} \vec{m}'\right\}.
\]
In words, $d_\G(\vec{m}, \vec{m}')$ is the weight of a shortest path
from $\vec{m}$ to $\vec{m}'$ in the graph induced by the relaxed step
relation $\areach{\G}{}$, where weights are scaled accordingly.

We now show that any $d_\G$, which we call the \emph{$\G$-distance}, 
is a distance under-approximation, and first show
effectiveness of all $d_\G$. It is
well-known and readily seen that reachability over $\G \in \{\Z, \Q\}$
is characterized by the following \emph{state equation}, since
transitions are always firable due to the absence of guards:
\begin{align*}
  \vec{m} \greach{*} \vec{m}'
  &\iff \exists \vec{\sigma} \in \G_{\geq 0}^T :
  \vec{m}' = \vec{m} + \sum_{t \in T} \vec{\sigma}(t) \cdot \effect{t}.
\end{align*}
Here, $\vec{\sigma}$ can be seen as the Parikh image of a sequence
$\sigma$ leading from $\vec{m}$ to $\vec{m}'$.

\begin{restatable}{proposition}{PropDistUnguarded}\label{prop:dist:unguarded}
  The functions $d_\Z$, $d_\Q$, $d_{\Qnon}$ are effective.
\end{restatable}

\begin{proof}
  By the state equation, we have:
  \begin{align*}
    d_\G(\vec{m}, \vec{m}')
    &= \min\left\{\sum_{t \in T} \lambda(t) \cdot \vec{\sigma}(t) :
    \vec{\sigma} \in \G_{\geq 0}^T, \vec{m}' = \vec{m} + \sum_{t \in T}
  \vec{\sigma}(t) \cdot \effect{t}\right\}.
  \end{align*}
  Therefore, $d_\Q(\vec{m}, \vec{m}')$ (resp.\ $d_\Z(\vec{m},
  \vec{m}')$) are computable by (resp.\ integer) linear programming,
  which is is complete for \P\ (resp.\ \NP), in its variant where one
  must check whether the minimal solution is at most some bound.

  For $d_{\Qnon}$, note that the reachability relation of a continuous
  Petri net can be expressed in the existential fragment of linear
  real arithmetic~\cite{BFHH17}. Hence, effectiveness follows from the
  decidability of linear real arithmetic.\qed
\end{proof}

Altogether, we conclude that $d_\G$ is a distance
under-approximation. Furthermore, we can show that $d_\G$ yields
\emph{unbounded} heuristics, which, by~\Cref{thm:greedy:termination},
ensure termination of GBFS on reachable instances:

\begin{restatable}{theorem}{ThmHUnbounded}\label{prop:h:unbounded}\label{prop:relax:dom}
  Let $\G \in \{\Z, \Q, \Qnon\}$, then $d_\G$ is a distance
  under-approxi\-ma\-tion. Moreover, the heuristics arising from it are
  unbounded.
\end{restatable}

\begin{proof}
  Let $\pn = (P, T, f, \lambda)$ be a weighted Petri net.
  Effectiveness of $d_{\G}$ follows from
  \Cref{prop:dist:unguarded}. By definitions and a simple induction,
  ${\reach{\sigma}} \subseteq {\areach{\G}{\sigma}}$ for any sequence
  $\sigma \in T^*$, with weights left unchanged for unscaled
  transitions. This implies that $d_{\G}(\vec{m}, \vec{m}') \leq
  \dist_\pn(\vec{m}, \vec{m}')$ for every $\vec{m}, \vec{m}' \in
  \G^P$. Moreover, the triangle inequality holds since for every
  $\vec{m}, \vec{m}', \vec{m}'' \in \G^P$ and sequences $\sigma,
  \sigma'$:
  \[\vec{m} \areach{\G}{\sigma} \vec{m}'
  \areach{\G}{\sigma'} \vec{m}'' \text{ implies } \vec{m}
  \areach{\G}{\sigma\sigma'} \vec{m}''.\]

  Let us sketch the proof of the second part. Let
  $\vec{m}_\text{target}$ be a marking and let $h_\G$ be the heuristic
  obtained from $d_\G$ for $\vec{m}_\text{target}$. Since
  $h_{\Q}(\vec{m}) \leq h_{\G}(\vec{m})$ for all $\vec{m}$ and
  $\G \in \{\Z, \Qnon\}$, it suffices to prove that $d_\Q$ is
  unbounded. Suppose it is not. There exist $b \in \Qnon$ and pairwise
  distinct markings $\vec{m}_0, \vec{m}_1, \ldots$ each with
  $h_{\Q}(\vec{m}_i) \leq b$. Let $\vec{x}_i$ be a solution to the
  state equation that gives $h_{\Q}(\vec{m}_i)$. By
  well-quasi-ordering and pairwise distinctness, there is a
  subsequence such that $\vec{m}_{i_0}(p) < \vec{m}_{i_1}(p) < \cdots$
  for some $p \in P$. Thus,
  $\lim_{j \to \infty} \vec{m}_\text{target}(p) - \vec{m}_{i_j}(p) =
  -\infty$, and hence $\lim_{j \to \infty} \vec{x}_{i_j}(s) = \infty$
  for some $s \in T$ with $\effect{s}(p) < 0$. This means that $b \geq
  h_{\Q}(\vec{m}_{i_j}) = \sum_{t \in
  T} \lambda(t) \cdot \vec{x}_{i_j}(t) > b$ for a sufficiently large
  $j$.\qed
\end{proof}

\subsection{Directed reachability based on distance under-approximations}%
\label{ssec:directed:safety}

We have all the ingredients to use Algorithm~\ref{alg:directed-search} for
answering reachability queries.

A \emph{distance under-approximation scheme} is a mapping
$\mathcal{D}$ that associates a distance under-approximation
$\mathcal{D}(\pn)$ to each weighted Petri net $\pn$. Let
$h_{\mathcal{D}(\pn), \vec{m}_\text{target}}$ be the heuristic
obtained from $\mathcal{D}(\pn)$ for marking
$\vec{m}_\text{target}$. By instantiating
Algorithm~\ref{alg:directed-search} with this heuristic, we can search
for a short(est) firing sequence witnessing that
$\vec{m}_\text{target}$ is reachable. Of course, constructing the
reachability graph of $\pn$ would be at least as difficult as
answering this query, or impossible if it is infinite. Hence, we
provide $G_\N(\pn)$ \emph{symbolically} through $\pn$ and let
Algorithm~\ref{alg:directed-search} explore it on-the-fly by
progressively firing its transitions.

For each $\G \in \{\Z, \Q, \Qnon\}$, the function $\mathcal{D}_\G$
mapping a weighted Petri net $\pn$ to its $\G$-distance $d_\G$ is a
distance under-approximation scheme with consistent and unbounded
heuristics by \Cref{prop:h:consistent}, \Cref{thm:greedy:termination} and
\Cref{prop:h:unbounded}. Although Algorithm~\ref{alg:directed-search} is geared towards finding
paths, it can prove \emph{non}-reacha\-bi\-li\-ty even for infinite
reachability graphs. Indeed, at some point, every candidate marking
$\vec{m} \in C$ may be such that
$h_{\mathcal{D}(\pn), \vec{m}_\text{target}}(\vec{m}) = \infty$, which
halts with $\infty$. There is no guarantee that this happens, but, as
reported
\eg\ by~\cite{ELMMN14,BFHH17}, the $\G$-distance for domains $\G \in
\{\Z, \Q, \Qnon\}$ does well for witnessing non-reachability in
practice, often from the very first marking
$\vec{m}_\text{init}$.

\paragraph*{An example.}

We illustrate our approach with a toy example and $\mathcal{D}_\Q$
(the scheme based on the state equation over $\Qnon^T$). Consider the
Petri net $\pn$ illustrated on the left of \Cref{fig:search:ex}, but
marked with $\vec{m}_\text{init} \defeq [p_1 \colon 0, p_2 \colon
0]$. Suppose we wish to determine whether $\vec{m}_\text{init}$ can
reach marking $\vec{m}_\text{target} \defeq [p_1 \colon 0, p_2 \colon
1]$ in $\pn$.

We consider the case where Algorithm~\ref{alg:directed-search} follows a greedy
best-first search, but the markings would be expanded in the same way
with \astar. Let us abbreviate a marking $[p_1 \colon x, p_2 \colon
  y]$ as $(x, y)$. Since $\effect{t_2} = (0, 1)$, the heuristic
considers that $\vec{m}_\text{init}$ can reach $\vec{m}_\text{target}$
in a single step using transition $t_2$ (it is unaware of the
guard). Marking $(1, 0)$ is expanded and its heuristic value increases
to $2$ as the state equation considers that both $t_2$ and $t_3$
must be fired (in some unknown order). Markings $(2, 0)$ and $(1, 1)$
are both discovered with respective heuristic values $3$ and $1$. The
latter is more promising, so it is expanded and target $(0, 1)$ is
discovered. Since its heuristic value is $0$, it is immediately
expanded and the correct distance $\dist_\pn(\vec{m}_\text{init},
\vec{m}_\text{target}) = 3$ is returned. Note that, in this example,
the only markings expanded are precisely those occurring on the
shortest path.

\paragraph*{Handling multiple targets.}

Algorithm~\ref{alg:directed-search} can be adapted to search for
\emph{some} marking from a given target set $X \subseteq \N^P$. The
idea consists simply in using a heuristic $h_X \colon \N^P \to \Qnon
\cup \{\infty\}$ estimating the weight of a shortest path to
\emph{any} target: \[h_X(\vec{m}) \defeq
\min\{h_{\mathcal{D}(\pn), \vec{m}_\text{target}}(\vec{m}) :
\vec{m}_\text{target} \in X\}.\]
This is convenient for partial reachability instances occurring in
practice, \ie\
\[
X \defeq \left\{\vec{m}_\text{target} \in \N^P \colon
\vec{m}_\text{target}(p) \sim_p \vec{c}(p)\right\} \text{ where }
\vec{c} \in \N^P \text{ and } each \sim_p \in \{=, \geq\}.
\]

\section{Experimental results}
\label{sec:experimental}
We implemented Algorithm~\ref{alg:directed-search} in a prototype,
called \textsc{FastForward}, which supports all selection strategies
and distance under-approxi\-ma\-tions presented in the paper. We
evaluate \textsc{FastForward} empirically with three main goals in
mind. First, we show that our approach is competitive with established
tools and can even vastly outperform them, and we also give insights
on its performance w.r.t.\ its parameterizations. Second, we compare
the length of the witnesses reported by the different tools. Third, we
briefly discuss the quality of the heuristics.

\subsubsection*{Technical details.}

Our tool is written in \textsc{C\#} and uses
\textsc{Gurobi}~\cite{Gur20}, a state-of-the-art MILP solver, for
distance under-approxi\-ma\-tions. We performed our benchmarks on a
machine with an 8-Core Intel\textregistered~Core\texttrademark~i7-7700
CPU @ 3.60GHz running on Ubuntu 18.04 with memory constrained to
$\sim$8GB\@. We used a timeout of 60 seconds per instance, and all
tools were invoked from a \textsc{Python} script using the \code{time}
module for time measurements.

A minor challenge arises from the fact that many instances specify an
upward-closed set of initial markings rather than a single one. For
example, $\vec{m}_\text{init}(p) \geq 1$ to specify, \eg, an arbitrary
number of threads. We handle this by setting $\vec{m}_\text{init}(p) =
1$ and adding a transition $t_p$ producing a token into $p$.

As a preprocessing step, we implemented \emph{sign
analysis}~\cite{GLS18}. It is a general pruning technique that has
been shown beneficial for reducing the size of the state-space of
Petri nets. Initially, places that carry tokens are viewed as
marked. For each transition whose input places are marked, the output
places also become marked. When a fixpoint is reached, places left
unmarked cannot carry tokens in any reachable marking, so they are
discarded.

\label{ssec:exp:implem}

\subsubsection*{Benchmarks.}

Due to the lack of tools handling reachability for \emph{unbounded}
state spaces, benchmarks arising in the literature are primarily
\emph{coverability} instances\footnote{The Model Checking Contest
  focuses on reachability for \emph{finite} state spaces.},
\ie\ reachability towards an upward closed set of target markings. We
gathered 61 positive and 115 negative coverability instances
originating from five suites~\cite{KKW14,P02,BMDS07,J09,DKO13}
previously used for benchmarking~\cite{ELMMN14,BFHH17,GLS18}. They
arise from the analysis of multi-threaded \textsc{C} programs with
shared-memory; mutual exclusion algorithms; communication protocols;
provenance analysis in the context of a medical messaging and a
bug-tracking system; and the verification of \textsc{Erlang}
concurrent programs. We further extracted the \emph{sypet} suite made
of 30 positive (standard) reachability instances arising from queries
encountered in type-directed program synthesis~\cite{FMWDR17}. The
overall goal of this work is to enable a vast range of untapped
applications requiring reachability over unbounded state-spaces,
rather than just coverability. To obtain further (positive) instances
of the Petri net reachability problem, we performed random walks on
the Petri nets from the aforementioned coverability benchmarks. To
this end, we used the largest quarter of distinct Petri nets from each
coverability suite, for a total of 33. We performed one random walk
each of lengths 20, 25, 30, 35, 40, 50, 60, 75, 90 and 100, and we
saved the resulting marking as the target. For nets with an
upward-closed initial marking, we randomly chose to start with a
number of tokens between 1 and 20\% of the length of the walk. It is
important to note that even with long random walks, instances can (and
in fact tend to) have short witnesses. To remove trivial instances and
only keep the most challenging ones, we removed those instances where
\textsc{FastForward} or \textsc{LoLA} reported a witness of length at
most 20, disregarding the transitions used to generate the initial
marking. This leaves us with 127 challenging instances on which the
shortest witness is either unknown or has length more than
20. Moreover, this yields real-world Petri nets with no bias towards
any specific kind of targets.

This table summarizes the characteristics of the various benchmarks:

\begin{center}
  \setlength{\tabcolsep}{4pt}
  \small
  \begin{tabular}{crrrrrrrrr}
    \toprule
    \multirow{2}{*}{\textbf{Suite}} &
    \multirow{2}{*}{\textbf{Size}} &
    \multicolumn{4}{c}{\textbf{Number of places}} &
    \multicolumn{4}{c}{\textbf{Number of transitions}} \\

    &&
    min.\ & med.\ & mean & max.\ &
    min.\ & med.\ & mean & max.\ \\

    \midrule
    \textsc{coverability} &
    61 &
    16 &  87 &  226 &  2826 &
    14 & 181 & 1519 & 27370 \\

    \midrule
    \textsc{sypet} &
     30 &
     65 &  251 &  320 & 1199 &
    537 & 2307 & 2646 & 8340 \\

    \midrule
    \textsc{random walks} &
    127 &
     52 &  306 & 531 & 2826 &    
     60 & 3137 & 5885 & 27370 \\

    \bottomrule
  \end{tabular}
\end{center}

\subsubsection*{Tool comparison.}

To evaluate our approach on reachability instances, we compare
\textsc{FastForward} to \textsc{LoLA}~\cite{K02}, a tool developed for
two decades that wins several categories of the Model Checking Contest
every year. \textsc{LoLA} is geared towards model checking of finite
state spaces, but it implements semi-decision procedures for the
unbounded case. We further compare the three selection strategies of
Algorithm~\ref{alg:directed-search}: \astar, GBFS and Dijkstra; the
two first with the distance under-approximation scheme
$\mathcal{D}_{\Q}$, which provides the best trade-off between estimate
quality and efficiency. We also considered comparing with
\textsc{KReach}~\cite{DL20}, a tool showcased at TACAS'20 that
implements an exact non-elementary algorithm. However, it timed out on
 all instances, even with larger time limits.

\begin{figure}[h]
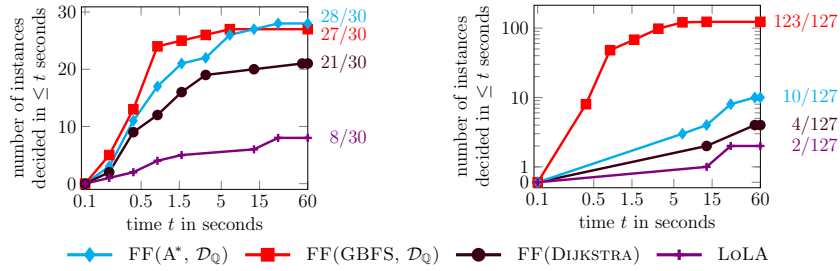

  \centering  
  \begin{minipage}[t]{0.45\textwidth}%
      \vspace*{0pt}%
      \input{benchmark-figs/fig-sypet-not-prepruned.tex}%
    \end{minipage}%
  \hspace*{15pt}%
  \begin{minipage}[t]{0.45\textwidth}%
      \vspace*{0pt}%
      \input{benchmark-figs/fig-reachability-not-prepruned.tex}%
  \end{minipage}%
  
  \begin{tikzpicture}[very thick, scale=\plotscale,
      every node/.style={scale=\plotscale}]
    \begin{customlegend}[legend columns=4,
        legend style={draw=none, column sep=1ex},
        legend entries={
          \textsc{FF(\astar, $\mathcal{D}_{\Q}$)},
          \textsc{FF(GBFS, $\mathcal{D}_{\Q}$)},
          \textsc{FF(Dijkstra)},
          \textsc{LoLA}},]
      \addlegend{FastForward QMarkingEQGurobi+Competitive}{diamond*}
      \addlegend{GBFS}{square*}
      \addlegend{FastForward zero+Competitive}{*}
      \addlegend{Lola}{+}
    \end{customlegend}
  \end{tikzpicture}
  \caption{Cumulative number of reachability instances decided over
    time. \emph{Left}: \textsc{sypet} suite (semi-log
    scale). \emph{Right}: \textsc{random-walk} suite (log
    scale).}\label{fig:reach:all}
\end{figure}

\Cref{fig:reach:all} depicts the number of reachability instances
decided by the tools within the time limit. As shown, all
approaches outperform \textsc{LoLA}, with GBFS as the clear winner on
the \textsc{random-walk} suite and \astar\ slightly better on the
\textsc{sypet} suite. Note that Dijkstra's selection strategy
sometimes competes due to its locally very cheap computational cost (no
heuristic evaluation), but its performance generally decreases as the
distance increases.

To demonstrate the versatility of our approach, we also benchmarked
\textsc{FastForward} on the original coverability instances. Recall
that coverability is an \EXPSPACE-complete problem that reduces to
reachability in linear time~\cite{Lip76,Rac78}. While its complexity
exceeds the \PSPACE-complete\-ness of reachability for finite
state-spaces~\cite{JLL77,Esp98}, it is much more tame than the
non-elementary complexity of (unbounded) reachability. We compare
\textsc{FastForward} to four tools implementing algorithms tailored
specifically to the coverability problem: \textsc{LoLA},
\textsc{Bfc}~\cite{KKW14}, \textsc{ICover}~\cite{GLS18} and the
backward algorithm (based on~\cite{ACJT96}) of
\textsc{mist}~\cite{P02}. We did not test
\textsc{Petrinizer}~\cite{ELMMN14} since it only handles negative
instances, while we focus on positive ones; likewise for
\textsc{QCover}~\cite{BFHH17} since it is superseded by
\textsc{ICover}.

\begin{figure}[h]
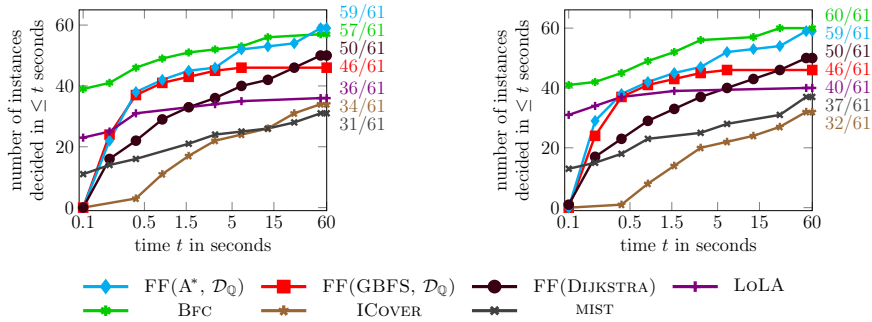

  \begin{minipage}[t]{0.48\textwidth}
    \input{benchmark-figs/fig-coverability-unsafe.tex}
  \end{minipage}%
  \hspace*{15pt}
  \begin{minipage}[t]{0.48\textwidth}
    \input{benchmark-figs/fig-coverability-prepruned.tex}
  \end{minipage}
  \vspace*{-15pt}
      \begin{center}
      \begin{tikzpicture}[very thick, scale=\plotscale,
        every node/.style={scale=\plotscale}]
\begin{customlegend}[legend columns=4,
legend style={draw=none, column sep=1ex},
legend entries={%
\textsc{FF(\astar, $\mathcal{D}_{\Q}$)},
\textsc{FF(GBFS, $\mathcal{D}_{\Q}$)},
\textsc{FF(Dijkstra)},
\textsc{LoLA},
\textsc{Bfc},
\textsc{ICover},
\textsc{mist}}]
\addlegend{FastForward QMarkingEQGurobi+Pruning+Competitive}{diamond*}
\addlegend{GBFS}{square*}
\addlegend{FastForward zero+Pruning+Competitive}{*}
\addlegend{Lola}{+}
\addlegend{BFC}{asterisk}
\addlegend{ICover}{star}
\addlegend{MIST}{x}
\end{customlegend}
\end{tikzpicture}
\end{center}
      \caption{Cumulative number of (positive) coverability instances
        decided over time. \emph{Left}: Evaluation on the original
        instances. \emph{Right}: Evaluation on the pre-pruned
        instances.}\label{fig:coverability}
\end{figure}

\Cref{fig:coverability} illustrates the number of coverability
instances decided within the time limit. The left side corresponds
to an evaluation on the original instances where \textsc{FastForward}
performs pruning (included in its runtime).
On the right hand right side the pruned instances are the
input for all tools, and the time for this pruning is not included for any
tool. As a caveat,
\textsc{ICover} performs its own preprocessing which includes pruning
among techniques specific to coverability. This preprocessing is
enabled (and its time is included) even when pruning is already done.
Using \textsc{FastForward(\astar, $\mathcal{D}_{\Q}$)}, we decide more
instances than all tools on unpruned Petri nets, and one less than
\textsc{Bfc} for pre-pruned instances. It is worth mentioning that
with a time limit of 10 minutes per instance, \textsc{FastForward(\astar,
  $\mathcal{D}_{\Q}$)} is the only tool to decide all 61 instances.

\begin{figure}[h]
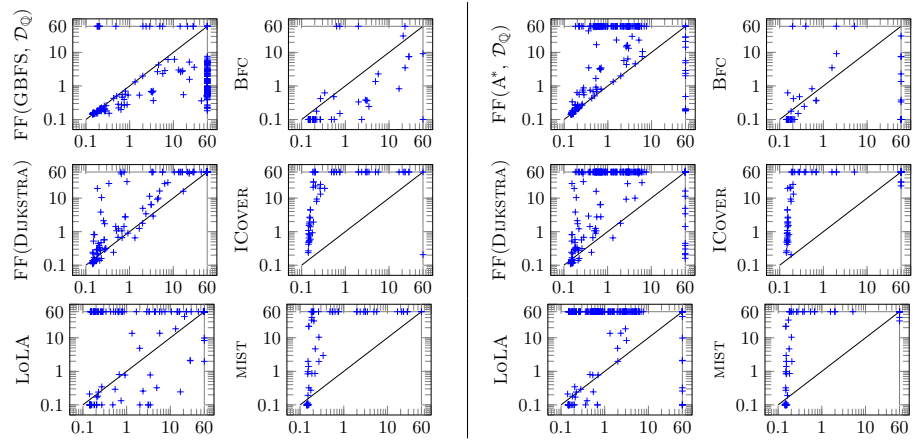

  \centering%
  \hspace*{-10pt}%
  \begin{minipage}[t]{0.29\textwidth}%
    \centering
    \input{benchmark-figs/fig-comparison-gbfs.tex}%
  \end{minipage}%
  \hspace*{-20pt}
  \begin{minipage}[t]{0.29\textwidth}%
    \centering
    \input{benchmark-figs/fig-comparison-bfc.tex}%
  \end{minipage}%
  \vrule
  \hspace*{-5pt}
  \begin{minipage}[t]{0.29\textwidth}%
    \centering
    \input{benchmark-figs/fig-gbfs-comparison-gbfs.tex}%
  \end{minipage}%
  \hspace*{-20pt}
  \begin{minipage}[t]{0.29\textwidth}%
    \centering
    \input{benchmark-figs/fig-gbfs-comparison-bfc.tex}%
  \end{minipage}%
  \vspace*{-1pt}

  \hspace*{-10pt}%
  \begin{minipage}[t]{0.29\textwidth}%
    \centering
    \input{benchmark-figs/fig-comparison-dijkstra.tex}%
  \end{minipage}%
  \hspace*{-20pt}
  \begin{minipage}[t]{0.29\textwidth}%
    \centering
    \input{benchmark-figs/fig-comparison-icover.tex}%
  \end{minipage}%
  \vrule
  \hspace*{-5pt}
  \begin{minipage}[t]{0.29\textwidth}%
    \centering
    \input{benchmark-figs/fig-gbfs-comparison-dijkstra.tex}%
  \end{minipage}%
  \hspace*{-20pt}
  \begin{minipage}[t]{0.29\textwidth}%
    \centering
    \input{benchmark-figs/fig-gbfs-comparison-icover.tex}%
  \end{minipage}%
  \vspace*{-1pt}

  \hspace*{-10pt}%
  \begin{minipage}[t]{0.29\textwidth}%
    \centering
    \input{benchmark-figs/fig-comparison-lola.tex}%
  \end{minipage}%
  \hspace*{-20pt}
  \begin{minipage}[t]{0.29\textwidth}%
    \centering
    \input{benchmark-figs/fig-comparison-mist.tex}%
  \end{minipage}%
  \vrule
  \hspace*{-5pt}
  \begin{minipage}[t]{0.29\textwidth}%
    \centering
    \input{benchmark-figs/fig-gbfs-comparison-lola.tex}%
  \end{minipage}%
  \hspace*{-20pt}
  \begin{minipage}[t]{0.29\textwidth}%
    \centering
    \input{benchmark-figs/fig-gbfs-comparison-mist.tex}%
  \end{minipage}%
  \vspace*{-1pt}
  
\caption{Runtime comparison against \textsc{FF(\astar,
    $\mathcal{D}_{\Q}$)} (\emph{left}) and \textsc{FF(GBFS,
    $\mathcal{D}_{\Q}$)} (\emph{right}), in seconds, for individual
  instances without pre-pruning. Tools on the first column of each
  side include coverability and reachability instances, while those on
  the second column of each side include coverability only. Marks on
  the \textcolor{gray!80!black}{gray} lines denote timeouts
  ($60$~s).}\label{fig:comparisons}
\end{figure}

We also compared the running time of \astar\ and GBFS with
$\mathcal{D}_{\Q}$ to the other tools and approaches. For each tool,
we considered the type of instances it can handle: either reachability
and coverability, or coverability only. \Cref{fig:comparisons} depicts
this comparison, where the base approach is faster for data points
that lie in the upper-left half of the graph. The axes start at 0.1
second to avoid a comparison based on technical aspects such as the
programming language. Yet, \textsc{LoLA}, \textsc{Bfc} and
\textsc{mist} regularly solve instances faster than this, which speaks
to their level of optimization. We can see that \textsc{FastForward}
outperforms \textsc{ICover}, \textsc{LoLA} and \textsc{mist}
overall. We cannot compete with \textsc{Bfc} in execution time as it
is a highly optimized tool specifically tailored to only the
coverability problem that can employ optimization techniques such as
Karp-Miller trees that do not work for reachability queries.

\subsubsection*{Length of the witnesses.}

Since our approach is also geared towards the identification of
short(est) reachability witnesses, we compared the different tools
with respect to length of the reported one, depicted in
\Cref{fig:path-comparison}. Positive values on the $y$-axis mean the
witness was not minimal, while $y = 0$ means it was. Note that the
points for \textsc{Bfc} must be taken with a grain of salt: it uses a
different file format, and its translation utility can introduce
additional transitions. This means that even if \textsc{Bfc} found a
shortest witness, it could be longer than a shortest one of the
original instance.

\begin{figure}

  \begin{center}
  {%
  \begin{tikzpicture}[scale=\plotscale, every node/.style={scale=\plotscale}, font=\large]%
      \begin{axis}[%
          xmin=-0.5,
          xmax=50.5,
          ymin=-0.5,
          height=5cm,
          log basis y={2},
          ytick={0,0.5,1,2,3,4,5,6,7},
          yticklabels={0,1,2,4,8,16,32,64},
          width=\textwidth,
          xlabel style={align=center},
          xlabel={length of the shortest witness},
          ylabel style={align=center},
          ylabel={difference with \\ exact distance},
          legend style={
            at={(0, 0.65)},
            anchor=west},
          reverse legend,
        ]%

        \addplot[only marks, color={colFastForward GBFS QMarkingEQGurobi+Pruning+Competitive}, mark=square*] coordinates {(13, 0) 
        (13, 0) 
        (14, 1.0) 
        (13, 0) 
        (4, 0) 
        (8, 0) 
        (36, 1.0) 
        (29, 0) 
        (17, 2.0) 
        (9, 0) 
        (12, 0) 
        (7, 0) 
        (30, 0) 
        (16, 0) 
        (17, 2.0) 
        (9, 1.0) 
        (14, 0) 
        (14, 0) 
        (16, 0) 
        (13, 0) 
        (18, 1.0) 
        (19, 0) 
        (22, 0) 
        (19, 0) 
        (10, 0) 
        (32, 0) 
        (10, 0) 
        (21, 0) 
        (26, 0) 
        (10, 0) 
        (12, 0) 
        (15, 1.0) 
        (9, 1.0) 
        (8, 0) 
        (19, 0) 
        (12, 0) 
        (10, 0) 
        (8, 0) 
        (12, 0) 
        (13, 0) 
        (13, 0) 
        (10, 0) 
        (12, 0) 
        (16, 1.0) 
        (4, 0) 
        (1, 0) 
        (2, 0) 
        (2, 0.5) 
        (2, 0) 
        (2, 0) 
        (3, 0) 
        (4, 0) 
        (4, 0.5) 
        (2, 0) 
        (3, 0) 
        (2, 0.5) 
        (2, 0) 
        (4, 0) 
        (3, 0) 
        (3, 0.5) 
        (1, 0) 
        (4, 0) 
        (3, 0.5) 
        (2, 0) 
        (1, 0) 
        (5, 1.0) 
        (2, 0.5) 
        (2, 0) 
        (3, 1.0) 
        (1, 0) 
        (3, 0.5) 
        (30, 3.3219280948873626) 
        (25, 0) 
        (50, 4.321928094887363) 
        (50, 0) 
        (25, 0) 
        (50, 0) 
        (50, 0) 
        (25, 3.3219280948873626) 
        (25, 0) 
        (40, 3.3219280948873626) 
        (25, 0) 
        (30, 3.3219280948873626) 
        };
                
        \addplot[only marks, color={colBFC}, mark=asterisk, very thick] coordinates {(13, 0.5) 
        (14, 3.700439718141092) 
        (13, 0) 
        (4, 0) 
        (8, 0) 
        (36, 0) 
        (17, 3.700439718141092) 
        (9, 0) 
        (12, 0.5) 
        (7, 3.9068905956085187) 
        (30, 1.0) 
        (16, 0) 
        (17, 3.700439718141092) 
        (9, 1.0) 
        (14, 2.0) 
        (14, 0) 
        (16, 1.0) 
        (13, 0.5) 
        (18, 3.9068905956085187) 
        (19, 1.0) 
        (22, 0) 
        (19, 1.0) 
        (10, 4.08746284125034) 
        (32, 2.0) 
        (10, 0) 
        (21, 0) 
        (26, 0) 
        (10, 0) 
        (12, 0) 
        (15, 2.584962500721156) 
        (9, 1.0) 
        (8, 1.0) 
        (19, 3.0) 
        (12, 3.1699250014423126) 
        (10, 3.700439718141092) 
        (8, 3.700439718141092) 
        (13, 0.5) 
        (10, 4.700439718141093) 
        (12, 0.5) 
        (16, 2.584962500721156) 
        (22, 1.5849625007211563) 
        (24, 1.0) 
        (24, 2.584962500721156) 
        (24, 0.5) 
        (24, 2.321928094887362) 
        (35, 2.584962500721156) 
        (11, 4.700439718141093) 
        (18, 2.321928094887362) 
        (16, 3.9068905956085187) 
        (18, 3.5849625007211565) 
        (19, 2.584962500721156) 
        (11, 3.8073549220576037) 
        (15, 3.8073549220576037) 
        (14, 4.392317422778761) 
        (17, 2.0) 
        };

        \addplot[mark=none] coordinates {(0,0)(50.5,0)};
        \addplot[only marks, color={colLoLA}, mark=+, thick] coordinates {(13, 0) 
        (13, 0) 
        (13, 0) 
        (4, 0) 
        (8, 0) 
        (17, 0) 
        (9, 0) 
        (12, 0) 
        (16, 0) 
        (17, 0) 
        (14, 0) 
        (14, 0) 
        (13, 0) 
        (18, 0) 
        (19, 0) 
        (10, 0) 
        (10, 0) 
        (26, 1.0) 
        (12, 0) 
        (15, 2.321928094887362) 
        (8, 0) 
        (10, 0) 
        (8, 1.5849625007211563) 
        (12, 0) 
        (13, 0) 
        (12, 0) 
        (24, 0.5) 
        (24, 3.3219280948873626) 
        (35, 4.0) 
        (18, 3.0) 
        (16, 0) 
        (18, 5.7279204545632) 
        (15, 3.5849625007211565) 
        (14, 2.584962500721156) 
        (17, 0) 
        (2, 2.584962500721156) 
        (3, 0) 
        (4, 1.0) 
        (1, 0.5) 
        (2, 3.1699250014423126) 
        (1, 0.5) 
        (1, 0) 
        (40, 0) 
        };
        
      \end{axis}%
    \end{tikzpicture}%
    }%
    \\
    \begin{tikzpicture}[very thick, scale=\plotscale,
      every node/.style={scale=\plotscale}]%
    \begin{customlegend}[legend columns=6,
    legend style={draw=none, column sep=1ex},
    legend entries={
\textsc{FF(GBFS, $\mathcal{D}_{\Q}$)}
,
\textsc{LoLA}
, 
\textsc{Bfc}
}]%

\addlegendmark{GBFS}{square*}

\addlegendmark{LoLA}{+}

\addlegendmark{BFC}{asterisk}

    \end{customlegend}
    \end{tikzpicture}
  \end{center}
    \caption{Length of the returned witness, per tool, compared to the
      length of a shortest witness. \textsc{ICover} is left out as it
      does not return witnesses. \textsc{FF(\astar, $\mathcal{D}_{\Q}$)}, \textsc{FF(Dijkstra)} and \textsc{mist} are left out
      as they are guaranteed to return shortest witnesses.}\label{fig:path-comparison}
    
\end{figure}
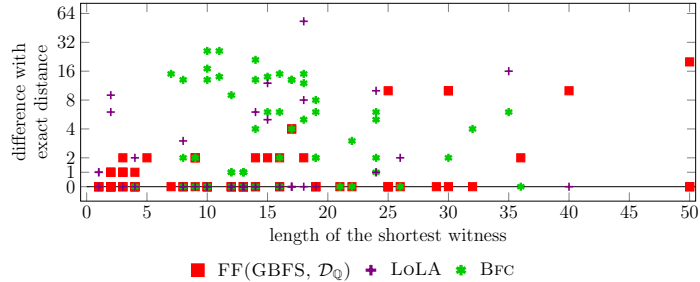

Still, the graph shows that reported witnesses can be far from
minimal. For example, on one instance \textsc{LoLA} returns a witness
that is 53 transitions longer than the one of
\textsc{FastForward(\astar, $\mathcal{D}_{\Q}$)}. Still, \textsc{LoLA}
returns a shortest witness on 28 out of 43 instances. Similarly,
\textsc{FastForward(GBFS, $\mathcal{D}_{\Q}$)} finds a shortest path
on 60 out of 83 instances\footnote{These numbers disregard instances
  where the tool did not finish or where a shortest witness is not
  known, \ie\ no method guaranteeing one finished in time.}.  In
contrast, \textsc{mist} finds a shortest witness on all instances
since its backward algorithm is guaranteed to do so on unweighted
Petri nets, which constitute all of our instances. Again, this
approach is tailored to coverability and cannot be lifted to
reachability.

\subsubsection*{Heuristics and pruning.}

We briefly discuss the quality of the heuristics and the impact of
pruning. The
left-hand side of \Cref{fig:other} compares the exact distance to the
estimated distance from the initial marking.\footnote{\textsc{Z3}
  reported two non optimal solutions which explains the two points
  above the line.} It shows that it is incredibly accurate for all
$\G$-distances, but even more so for $\G = \Q_{\geq 0}$. We
experimented with this distance using the logical translation
of~\cite{BFHH17} and \textsc{Z3}~\cite{MB08} as the optimization
modulo theories solver. At present, it appears that the gain in
estimate quality does not compensate for the extra computational cost.

As depicted on the right-hand side of \Cref{fig:other}, pruning can
make some instances trivial, but in general, many challenging
instances remain so. On average, around $50$\% of places and $40$\% of
transitions were pruned.

\colorlet{colNeq}{black}
\definecolor{colQeq}{HTML}{9393FF}
\colorlet{colCont}{magenta}
\begin{figure}[!h]
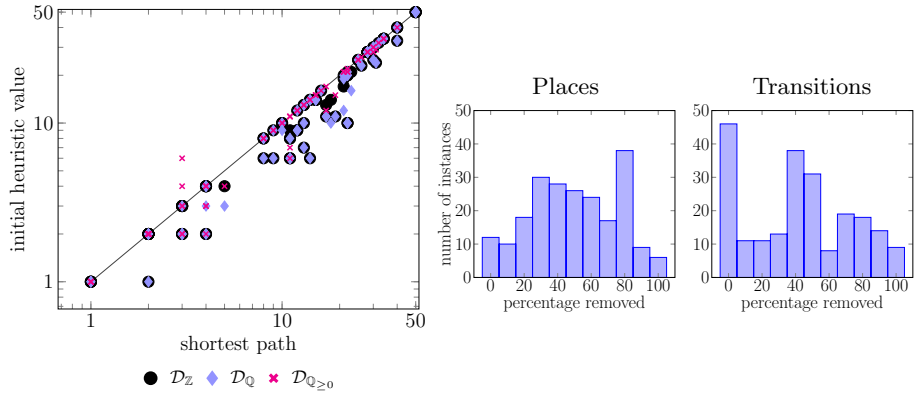

  \centering
  \begin{minipage}[c]{0.43\textwidth}
    \vspace{0pt}
    \heuristicplot{shortest path}{initial heuristic value}{

    \addplot[darkgray] plot coordinates {(1,1) (50, 50)};

    \addplot[only marks, mark=*, color=colNeq, mark size=3 pt] plot coordinates {(13, 13) (67, 67) (10, 10) (4, 4) (32, 32) (28, 28) (50, 50) (25, 25) (50, 50) (28, 28) (34, 34) (34, 34) (28, 28) (28, 28) (30, 30) (40, 40) (30, 30) (28, 28) (34, 34) (50, 50) (28, 28) (50, 50) (28, 28) (40, 40) (25, 25) (15, 14) (22, 20) (15, 14) (21, 20) (22, 20) (15, 14) (1, 1) (2, 2) (3, 3) (4, 4) (4, 2) (2, 2) (2, 2) (3, 2) (3, 3) (4, 4) (3, 3) (2, 2) (2, 2) (5, 4) (1, 1) (2, 2) (3, 3) (2, 2) (2, 2) (2, 1) (3, 3) (1, 1) (4, 2) (4, 4) (3, 2) (3, 2) (2, 2) (1, 1) (17, 13) (14, 6) (11, 9) (8, 8) (9, 6) (21, 19) (14, 14) (11, 8) (8, 8) (17, 11) (13, 10) (16, 16) (31, 24) (18, 14) (22, 10) (21, 17) (14, 14) (13, 13) (12, 12) (11, 6) (12, 9) (13, 7) (9, 9) (30, 25) (23, 21) (11, 8) (22, 10) (11, 6) (8, 6) (40, 33) (26, 23) (22, 21) (19, 11) (12, 9) (10, 10)};

    \addplot [only marks, mark=diamond*, color=colQeq, mark size=2.75pt] plot coordinates {(13, 13) (67, 67) (10, 10) (4, 4) (32, 32) (28, 28) (50, 50) (25, 25) (50, 50) (28, 28) (34, 34) (34, 34) (28, 28) (28, 28) (30, 30) (40, 40) (30, 30) (28, 28) (34, 34) (50, 50) (28, 28) (50, 50) (28, 28) (40, 40) (25, 25) (15, 14) (22, 20) (15, 14) (21, 20) (22, 20) (15, 14) (1, 1) (2, 2) (3, 2) (4, 4) (4, 2) (2, 2) (2, 2) (3, 2) (3, 2) (4, 3) (3, 3) (2, 2) (2, 2) (5, 3) (1, 1) (2, 2) (3, 2) (2, 2) (2, 2) (2, 1) (3, 2) (1, 1) (4, 2) (4, 4) (3, 2) (3, 2) (2, 2) (1, 1) (17, 11) (14, 6) (11, 8) (8, 8) (9, 6) (21, 19) (14, 14) (11, 6) (8, 8) (17, 11) (13, 10) (16, 16) (31, 24) (18, 10) (22, 10) (21, 12) (14, 14) (13, 13) (12, 12) (11, 6) (12, 9) (13, 7) (9, 9) (30, 25) (23, 16) (11, 6) (22, 10) (11, 6) (8, 6) (40, 33) (26, 23) (22, 21) (19, 11) (12, 9) (10, 9)};

    \addplot [only marks, mark=x, color=colCont, mark size = 2pt, thick] plot coordinates {(67, 67) (10, 10) (4, 4) (32, 32) (28, 28) (25, 25) (34, 34) (34, 34) (28, 28) (30, 30) (40, 40) (30, 30) (34, 34) (28, 28) (25, 25) (15, 15) (22, 21) (15, 15) (21, 21) (22, 21) (15, 15) (1, 1) (2, 2) (3, 2) (4, 3) (4, 4) (2, 2) (2, 2) (3, 4) (4, 3) (2, 2) (2, 2) (5, 4) (1, 1) (2, 2) (3, 6) (2, 2) (2, 2) (3, 2) (1, 1) (4, 2) (4, 4) (3, 3) (3, 3) (2, 2) (1, 1) (17, 12) (11, 11) (8, 8) (9, 9) (21, 21) (14, 14) (11, 6) (8, 8) (17, 17) (13, 13) (16, 16) (31, 29) (14, 14) (13, 13) (12, 12) (11, 11) (12, 12) (9, 9) (30, 28) (11, 7) (11, 11) (8, 8) (26, 26) (22, 22) (19, 15) (12, 12) (10, 10)};

    }

    \hspace*{30pt}\centering
\begin{tikzpicture}[very thick, scale=\plotscale,
    every node/.style={scale=\plotscale}]
\begin{customlegend}[legend columns=3,
    legend style={draw=none, column sep=1ex},
legend entries={
\textsc{$\mathcal{D}_\Z$},
\textsc{$\mathcal{D}_\Q$},
\textsc{$\mathcal{D}_{\Qnon}$},
  }] 
\addlegendmark{Neq}{*}
\addlegendmark{Qeq}{diamond*}
\addlegendmark{Cont}{x}
\end{customlegend}
\end{tikzpicture}
\end{minipage}%
\hspace*{10pt}
\begin{minipage}[c]{0.55\textwidth}
  \vspace{0pt}
  \begin{tabular}{cc}
    \footnotesize\, Places & \footnotesize\, Transitions \\
    \input{benchmark-figs/fig-pruning-places.tex} &
    \input{benchmark-figs/fig-pruning-transitions.tex}
  \end{tabular}
\end{minipage}

\caption{\emph{Left}: initial distance estimation compared to the
  exact distance (points closer to the diagonal are
  better). \emph{Right}: number of instances per percentage of places
  (left) and transitions (right) removed by pruning (rounded to
  nearest multiple of~10).}\label{fig:other}
\end{figure}


\section{Conclusion}
\label{sec:conclusion}
We presented an efficient approach to the Petri net reachability
problem that uses state-space over-approximations as distance oracles
in the classical graph traversal algorithms \astar\ and greedy
best-first search. Our experiments have shown that using the state
equation over $\Qnon^T$ provides the best trade-off between
computational feasibility and the accuracy of the oracle. However, we
expect that further advances in optimization modulo theories solvers
may enable employing stronger over-approximations such as continuous
Petri nets in the future.

Moreover, non-algebraic distance under-approximations also fit
naturally in our framework, \eg\ the syntactic distance of~\cite{S14}
and ``$\alpha$-graphs'' of~\cite{FMWDR17}. These are crude
approximations with low computational cost. Our preliminary tests show
that, although they could not compete with our distances, they can
provide early speed-ups on instances with large branching factors. An
interesting line of research consists in identifying cheap
approximations with better estimates.

We wish to emphasize that our approach to the reachability problem has
the potential to also be naturally used for semi-deciding reachabiltiy
in extensions of Petri nets with a recursively enumerable reachability
problem, such as Petri nets with resets and
transfers~\cite{AK76,DFS98} as well as colored Petri
nets~\cite{Jen13}. These extensions have, for instance, been used for
the generation of program loop invariants~\cite{SK19}, the validation
of business processes~\cite{WAHE09} and the verification of
multi-threaded \textsc{C} and \textsc{Java} program skeletons with
communication primitives~\cite{DRV02,KKW14}. Linear rational and
integer arithmetic over-approximations for such extended Petri nets
exist~\cite{CHH18,BHM18,HLT17,GSAH19} and could smoothly be used
inside our framework.


\section*{Acknowledgments}

We thank Juliette Fournis d'Albiat for her help with extracting the
\textsc{sypet} suite.

\bibliographystyle{splncs04}
\bibliography{references}

\appendix
\section{A primer on applications of Petri net reachability}

\label{sec:motivation}
This section provides two representative examples from the literature
that illustrate the important role of Petri net reachability. They
allow us to underpin our claim that it is desirable to find shortest
paths witnessing reachability. Our examples come from program
synthesis and concurrent program analysis. We conclude this section
with a brief discussion on further applications.

\paragraph*{Program synthesis.}

The authors of~\cite{FMWDR17} and~\cite{GJJZWJP20} have recently
employed the Petri net reachability problem for automated program
synthesis.  In their setting, one is given an API containing hundreds
or thousands of functions, together with a type signature and a number
of test cases. The goal is to automatically synthesize a loop-free
program using functions from the API that respects the specified type
signature and satisfies the given test cases.

\begin{figure}[!h]
  \centering
  \footnotesize
  \begin{tabular}{ |l| }
    \hline
    \multicolumn{1}{|c|}{\code{java.awt.geom}} \\
    \hline 
    \code{\codeword{new} AffineTransformation()}\\
    \code{Shape Shape.createTransformedShape(AffineTransformation)}\\
    \code{String Point2D.ToString()}\\
    \code{\codeword{double} Point2D.getX()}\\
    \code{\codeword{double} Point2D.getY()}\\
    \code{\codeword{void} AffineTransformation.setToRotation(\codeword{double}, \codeword{double}, \codeword{double}) }\\
    \code{\codeword{void} AffineTransformation.invert()}\\
    \code{Area Area.createTransformedArea(AffineTransformation)}\\
    \hline
  \end{tabular}    
  \caption{A small sample of methods from library \code{java.awt.geom}.}\label{fig:sypet-api}
\end{figure}

Let us illustrate the approach with an example
from~\cite{FMWDR17}. Suppose we have access to library
\code{java.awt.geom}, and we wish to synthesize a function
\code{rotate} with type signature
\begin{center}
  \code{Area rotate(Area object, Point2D point, double
    angle)}.
\end{center}
Naturally, the function should rotate the supplied \code{Area} around
\code{point} by \code{angle} degrees. We assume the
\code{java.awt.geom} library is sufficient for this task in that it
contains the functions needed to synthesize the
method. \Cref{fig:sypet-api} presents an excerpt of functions
contained in the API\@.

The authors of~\cite{FMWDR17} suggest to view an API as a Petri net
whose places correspond to types and transitions correspond to API
functions which, informally speaking, consume input types and produce
an output type. \Cref{fig:sypet-pn} illustrates the Petri net
corresponding to the excerpt of API functions listed in
\Cref{fig:sypet-api}. To synthesize the \code{rotate} function above,
we start with tokens in the places corresponding to the input
parameters of our function. Thus, in \Cref{fig:sypet-pn} we have one
token in each of the places corresponding to \code{Area},
\code{Point2D} and \code{double}. The goal is then to reach a marking
with a single token in the place corresponding to the return type. In
our example, we aim for one token in \code{Area}, and no token in any
other place. This corresponds to invoking a sequence of functions that
``use up'' all input parameters, and finally return the correct
type. To allow reuse of variables, additional ``copy'' transitions are
introduced for each place; they take one token from a place and put
two tokens back. If the target marking is reachable, then the
witnessing path corresponds to a partial sketch of a program.

For example, the path
\begin{multline*}
  \text{copy}_{\text{Point2D}} \rightarrow \text{GetY} \rightarrow \text{GetX}
  \rightarrow \text{new AffineTransformation} \rightarrow \\
  \text{copy}_{\text{AffineTransformation}} \rightarrow \text{setToRotation}
  \rightarrow \text{createTransformedArea} 
\end{multline*}
tells us which functions to apply, and in which order to apply them.
Since Petri nets do not store information about the identity of
tokens, when we have multiple objects of the same type, we do not know
which to supply as an argument to which function. This can be figured
out by a separate process involving SAT solving (see~\cite{FMWDR17}
for more details).

As discussed in~\cite{FMWDR17}, finding short paths of the Petri net
is a natural goal. Indeed, since short programs are easier to test,
there are fewer possibilities for the arguments of each function, and
it is easier for humans to verify that the synthesized program has the
desired functionality.

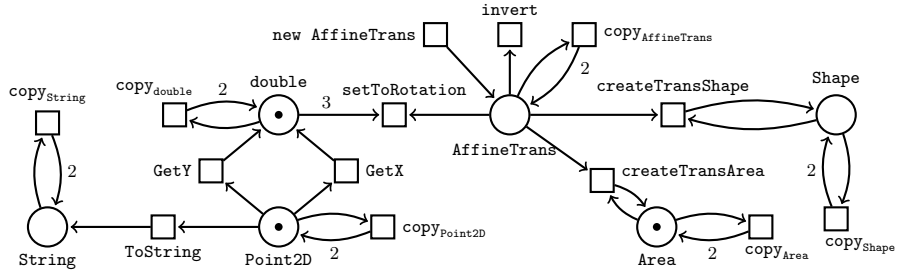
\begin{figure}[!h]
    \centering
    \begin{tikzpicture}[auto, thick, transform shape, scale=0.9,
                        every node/.style={scale=0.85}, font=\large,
                        font=\ttfamily]
      \node[place, label={[label distance=1pt, xshift=-2pt]267:{AffineTrans}}]
      (affine-transform) {};

      \node[place, label={Shape}, right=120pt of affine-transform] (shape) {};
      
      \node[below right=60pt and 88pt of affine-transform]
      (AreaHelperPoint) {};
      
      \node[] (AreaHelperMiddlePoint)
      at ($(affine-transform)!0.4!(AreaHelperPoint)$) {};
      
      \node[tokens=1, place, label=above:{double}, left=80pt of
        affine-transform] (double) {};
      
      \node[tokens=1, place, label=below:{Point2D}, below=30pt of
        double] (Point2D) {};
      
      \node[tokens=1, place, label=below:{Area}, xshift=27pt] (Area) at
      ($(AreaHelperMiddlePoint|-Point2D)$) {};
      
      \node[place, label=below:{String}, left=80pt of Point2D] (String) {};
            
      \node[transition, label={[xshift=5pt]above:{setToRotation}}]
      (setToRotation) at ($(affine-transform)!0.5!(double)$) {}
      edge [pre]                               (affine-transform)
      edge [pre] node[swap, xshift=-7pt] {$3$} (double);
      
      \node[transition, label=above:{invert}, above=20pt of affine-transform]
      (invert) {}
      edge [pre] (affine-transform);
      
      \node[transition, label=above:{createTransShape}]
      (createTransformedShape) at ($(affine-transform)!0.5!(shape)$) {}
      edge [pre]                 (affine-transform)
      edge [pre,  bend right=15] (shape)
      edge [post, bend  left=15] (shape);
      
      \node[transition, label=below:{ToString}] (Point2D ToString)
      at ($(Point2D)!0.5!(String)$) {}
      edge [pre]  (Point2D)
      edge [post] (String);
      
      \node[transition, label=right:{GetX}] (Point2D GetX)
      at ($(Point2D)!0.5!(double) + (1, 0)$) {}
      edge [pre]  (Point2D)
      edge [post] (double);
      
      \node[transition, label=left:{GetY}] (Point2D GetY)
      at ($(Point2D)!0.5!(double) + (-1, 0)$) {}
      edge [pre]  (Point2D)
      edge [post] (double);
      
      \node[transition, label={[yshift=5pt]right:{createTransArea}}]
      (CreateTransformedArea) at ($(AreaHelperMiddlePoint)$) {}
      edge [pre] (affine-transform)
      edge [pre,  bend right=20] (Area)
      edge [post, bend  left=20] (Area);
      
      \node[transition, label=left:{new AffineTrans},
            above left=22.5pt and 20pt of affine-transform] (setRotation) {}
      edge [post] (affine-transform);
      
      \node[transition, label=right:{$\text{copy}_{\text{AffineTrans}}$},
            above right=22.5pt and 20pt of affine-transform] (copy affine) {}
      edge [pre,  bend right=20] (affine-transform)
      edge [post, bend  left=20] node[label=above right:{$2$}] {}
      (affine-transform);
      
      \node[transition,
            label={[xshift=8pt]below:{$\text{copy}_{\text{Area}}$}},
            right=30pt of Area] (copy affine) {}
      edge [pre,  bend right=20]            (Area)
      edge [post, bend  left=20] node {$2$} (Area);
      
      \node[transition,
            label={[xshift=12pt]below:{$\text{copy}_{\text{Shape}}$}},
            below=30pt of shape] (copy affine) {}
      edge [pre,  bend right=20]            (shape)
      edge [post, bend  left=20] node {$2$} (shape);
      
      \node[transition,
            label={[xshift=-10pt]above:{$\text{copy}_\text{double}$}},
            left=30pt of double] (copy affine) {}
      edge [pre,  bend right=20]            (double)
      edge [post, bend  left=20] node {$2$} (double);
      
      \node[transition,
            label={[yshift=-2pt]right:{$\text{copy}_{\text{Point2D}}$}},
            right=30pt of Point2D] (copy affine) {}
      edge [pre,  bend right=20]              (Point2D)
      edge [post, bend  left=20] node[] {$2$} (Point2D);
      
      \node[transition, label=above:{$\text{copy}_{\text{String}}$},
            above=30pt of String] (copy affine) {}
      edge [pre,  bend right=20]            (String)
      edge [post, bend  left=20] node {$2$} (String);      
    \end{tikzpicture}
    \vspace*{-5pt} 
    \caption{A Petri net modelling the API of \Cref{fig:sypet-api}.}%
    \label{fig:sypet-pn}
\end{figure}

\paragraph*{Concurrent program analysis.}

Perhaps most prominently, Petri nets have been used in order to model
and analyze concurrent processes.  Let us begin with a simple example
illustrating how the Petri net reachability problem can be used in
order to detect race conditions in concurrent programs. Consider
function \code{fun()} of \Cref{fig:concur:ex} in which \code{s} is a
global shared Boolean variable. If there is a single thread running
\code{fun()}, then the condition of the \code{if}-statement in Line~3
never evaluates to true and an error cannot occur. However, if there
are two independently interleaved threads running \code{fun()}, it is
possible that one thread reaches Line~3 whilst \code{s} is set to 1,
which means an error could occur.

\begin{figure}
\begin{lstlisting}[language=python,firstnumber=0]
def fun():
  s = 1
  s = 0
  if s == 1: raise Err()
\end{lstlisting}
  \caption{Simple program with a potential race condition.}\label{fig:concur:ex}
\end{figure}

In more technical terms, we consider non-recursive Boolean programs in
which an unbounded number of identical programs run in parallel. The
authors of~\cite{GS92} showed that verifying safety properties of such
concurrent programs can be reduced to the \emph{coverability problem}
for Petri nets using a technique called counter abstraction. The
coverability problem is a weaker version of the reachability
problem. Given a target marking, the coverability problem asks whether
it is possible to reach a marking in which every place carries at
least as many tokens as specified by the target marking. The Petri net
obtained by applying the approach of~\cite{GS92} to the program from
\Cref{fig:concur:ex} is depicted in \Cref{fig:concur:pn:ex}. The
places on the top of the Petri net correspond to the program locations
of \Cref{fig:concur:ex}.  Tokens in each of the places on the top
count the number of threads which are currently at the respective
program location, which is a form of counter abstraction. At any time,
transition \code{fun()} can add tokens to \code{loc$_1$}, reflecting
that a new thread executing \code{fun()} can be spawned at any point
in time arbitrarily often. The two places on the bottom encode the
state of the Boolean variable \code{s} which is updated whenever a
transition moves tokens from \code{loc$_1$} to \code{loc$_2$}, or from
\code{loc$_2$} to \code{loc$_3$}. Determining whether an error can
occur then reduces to deciding whether the marking $[\code{Err} \colon
  1]$ is coverable, \ie, whether there is an interleaving in which at
least one thread produces an error.
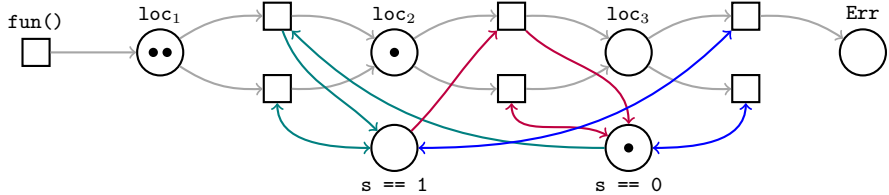
\begin{figure}
  \begin{tikzpicture}[auto, thick, node distance=1.25cm,
                      scale=0.9, transform shape]
    \newcommand{\hgap}{36pt}
    \newcommand{\vgap}{2pt}

    \node[transition]                                               (tin) {};
    \node[place,            right=          \hgap of tin, tokens=2] (p1)  {};
    \node[transition, above right=\vgap and \hgap of p1]            (t1l) {};
    \node[transition, below right=\vgap and \hgap of p1]            (t1r) {};
    \node[place,      below right=\vgap and \hgap of t1l, tokens=1] (p2)  {};
    \node[transition, above right=\vgap and \hgap of p2]            (t2l) {};
    \node[transition, below right=\vgap and \hgap of p2]            (t2r) {};
    \node[place,      below right=\vgap and \hgap of t2l]           (p3)  {};
    \node[transition, above right=\vgap and \hgap of p3]            (t3l) {};
    \node[transition, below right=\vgap and \hgap of p3]            (t3r) {};
    \node[place,      below right=\vgap and \hgap of t3l]           (p4)  {};

    \node[above=-1pt of p1]  {$\texttt{loc}_1$};
    \node[above=-1pt of p2]  {$\texttt{loc}_2$};
    \node[above=-1pt of p3]  {$\texttt{loc}_3$};
    \node[above=.5pt of p4]  {$\texttt{Err}$};
    \node[above= 0pt of tin] {$\texttt{fun()}$};

    \newcommand{\eangle}{18}
    
    \path[->, gray!70]
    (tin) edge[]                   node[] {} (p1)
    (p1)  edge[bend  left=\eangle] node[] {} (t1l)
    (p1)  edge[bend right=\eangle] node[] {} (t1r)
    (t1l) edge[bend  left=\eangle] node[] {} (p2)
    (t1r) edge[bend right=\eangle] node[] {} (p2)
    (p2)  edge[bend  left=\eangle] node[] {} (t2l)
    (p2)  edge[bend right=\eangle] node[] {} (t2r)
    (t2l) edge[bend  left=\eangle] node[] {} (p3)
    (t2r) edge[bend right=\eangle] node[] {} (p3)
    (p3)  edge[bend  left=\eangle] node[] {} (t3l)
    (p3)  edge[bend right=\eangle] node[] {} (t3r)
    (t3l) edge[bend  left=\eangle] node[] {} (p4)
    ;

    \newcommand{\sgap}{20pt}
    
    \node[place, below=\sgap of p2]           (s1) {};
    \node[place, below=\sgap of p3, tokens=1] (s0) {};

    \node[below=-1pt of s1] {$\texttt{s == 1}$};
    \node[below=-1pt of s0] {$\texttt{s == 0}$};

    \path[<->, teal]
    (s1)  edge[out=180, in=-90]  node[] {} (t1r)
    ;

    \path[->, teal]
    (s0)  edge[out=180, in=-45]  node[] {} (t1l)
    (t1l) edge[out=-70, in=135]  node[] {} (s1)
    ;

    \path[<->, purple]
    (s0)  edge[out=150, in=-90]  node[] {} (t2r)
    ;

    \path[->, purple]
    (s1)  edge[out= 45, in=-135] node[] {} (t2l)
    (t2l) edge[out=-45, in=  90] node[] {} (s0)
    ;

    \path[<->, blue]
    (s1)  edge[out=0, in=-135]  node[] {} (t3l)
    (s0)  edge[out=0, in= -90]  node[] {} (t3r)
    ;
  \end{tikzpicture}%
  \caption{The Petri net modeling the program of \Cref{fig:concur:ex}.
    A token in place $\texttt{loc}_i$ represents a thread at program
    location $i$, and a token in place \texttt{s == $b$} indicates
    that variable \texttt{s} has value $b$. Transition \texttt{fun()}
    spawns threads. A bidirectional arc $p \leftrightarrow t$
    abbreviates two arcs: $p \rightarrow t$ and $t \rightarrow
    p$. Colors are only meant to help
    readability.}\label{fig:concur:pn:ex}
\end{figure}

In stark contrast to the reachability problem, it was shown
in~\cite{Rac78} that the coverability problem belongs
to \EXPSPACE\@. There is a natural reduction from the coverability
problem to the reachability problem: by introducing additional
transitions that can non-deterministically remove tokens from every
place corresponding to program lines, a target marking in the original
Petri net is coverable iff it is reachable in the Petri net with the
additional transitions. Alternatively, deciding coverability can be
rephrased as the problem of determining whether an upward-closed set
of markings is reachable in the directed graph induced by a given
Petri net, which is the approach that we take.

\paragraph*{Further applications}

The authors of~\cite{FKP14} show how proofs involving counting
arguments, which can, for instance, naturally prove properties of
concurrent programs with recursive procedures, can automatically be
synthesized by a reduction to the Petri net reachability problem. The
authors of~\cite{GM12} propose a model for reasoning about finite-data
asynchronous programs. They show that proving liveness properties of
such programs in their model is inter-reducible with the Petri net
reachability problem. In a broader context, it was shown that various
verification problems for population protocols, a formal model of
sensor networks, reduce to the Petri net reachability
problem~\cite{EGLM17}. The authors of~\cite{DLV19} develop a method
that allows for verifying rich models of data-driven workflows by a
reduction to the coverability problem for Petri nets. See also
survey~\cite{Mur89} for further classical application areas of Petri
nets and their extensions.


\section{Missing proofs of \Cref{sec:search}}

Recall the following invariant satisfied by
Algorithm~\ref{alg:directed-search}:
\begin{multline}
  \text{if $g(v) \neq \infty$, then $g(v)$ is the weight of a path
    from $s$ to $v$ in $G$} \\[-3pt] \text{whose nodes were all
    expanded, except possibly $v$}.\tag*{($\ast$)}\label{eq:g:path2}
\end{multline}

We prove this lemma from the main text:

\LemGreedy*

\begin{proof}
  Let $d \defeq \min\{\mu(e) : e \in E\}$.
  \begin{enumerate}
  \item Any path of weight at most $c$ traverses at most $k \defeq \lceil c
    / d\rceil$ edges. Since the graph has finite out-degree, the
    number of paths from $s$ using at most $k$ edges is
    finite.\smallskip

  \item Suppose the claim false. We have $\dist_G(v_0, t) > \dist_G(v_1,
    t) > \cdots$ for some $v_0, v_1, \ldots \in W$. Let $k \defeq
    \lceil \dist_G(v_0, t) / d\rceil$. Let $V_{\leq k}$ be the set of
    nodes that can reach $t$ by traversing at most $k$ edges. Since $G$ has
    finite in-degree, $V_{\leq k}$ is finite. Moreover, any node $v
    \in V \setminus V_{\leq k}$ is such that $\dist_G(v, t) > k \cdot d
    \geq \dist_G(v_0, t)$. Hence, $\{v_0, v_1, \ldots\} \subseteq
    V_{\leq k}$ is finite, which is a contradiction.\smallskip

  \item For the sake of contradiction, assume a node $v$ is expanded
    infinitely often. Each time node $v$ is expanded, it is removed
    from $C$. Hence, it is reinserted infinitely often in
    $C$. Moreover, each time this happens, value $g(v)$ is
    decreased. Let $q_0, q_1, \ldots \in \Qnon$ denote these
    increasingly smaller values. By~\ref{eq:g:path2}, there is a path
    $\path_i$ from $s$ to $v$ of weight $q_i$ in
    $G$. By~\eqref{itm:path:fin}, $\{\path_i : i \in \N\}$ is finite
    as the weight of these paths is at most $q_0$. This contradicts
    $q_0 > q_1 > \cdots$. \qed
  \end{enumerate}
\end{proof}

\section{Missing proofs of \Cref{ssec:pn:relax}}

\PropDistUnguarded*

\begin{proof}
  Let us prove the case of $d_{\Qnon}$ which was only sketched in the
  main text. The reachability relation of a continuous Petri net can
  be expressed in the existential fragment of linear real arithmetic,
  \ie\ $\mathrm{FO}\langle \Q, +, < \rangle$, the first-order theory
  of the rationals with addition and order~\cite{BFHH17}. More
  precisely, there exists a linear-time computable formula $\psi \in
  \exists\, \mathrm{FO}\langle \Q, +, < \rangle$ such that
  $\psi(\vec{m}, \vec{x}, \vec{m}')$ holds iff
  \[\text{ there exists a sequence }
  \sigma \in ((0, 1] \times T)^* \text{ s.t.\ } \vec{m}
    \creach{\sigma} \vec{m}' \text{ and } \vec{\sigma} = \vec{x}.\]
  Let $\Phi(\vec{m}, \vec{m}', \ell) \defeq \exists \vec{x} \in
  \Qnon^T : \psi(\vec{m}, \vec{x}, \vec{m}') \land \ell = \sum_{t \in
    T} \lambda(t) \cdot \vec{x}(t)$. Formula $\Phi \in \exists\,
  \mathrm{FO}\langle \Q, +, < \rangle$ can be constructed in linear
  time and is such that $\Phi(\vec{m}, \vec{m}', \ell)$ holds for
  $\vec{m}, \vec{m}' \in \Qnon^P$ and $\ell \in \Qnon$ iff $\ell =
  d_{\Qnon}(\vec{m}, \vec{m}')$. Thus, $d_{\Qnon}$ is computable as an
  instance of a decidable optimization modulo theories problem.\qed
\end{proof}

\ThmHUnbounded*

\begin{proof}
  The first was part of the statement was fully shown in the main
  text. Let us prove the second part more formally. Let $\pn = (P, T,
  f, \lambda)$ be a weighted Petri net, let $\vec{m}_\text{target}$ be
  a target marking, and let $h_\G$ be the heuristic obtained from
  $d_\G$ for $\vec{m}_\text{target}$. Observe that $h_{\Q}(\vec{m})
  \leq h_{\G}(\vec{m})$ for every marking $\vec{m}$ and every $\G \in
  \{\Z, \Q, \Qnon\}$. Hence, if $h_{\Q}$ is unbounded, so are all
  three heuristics. Thus, it suffices to prove the case $\G = \Q$.

  For the sake of contradiction, suppose $h_{\Q}$ is not
  unbounded. There exists $b \in \Qnon$ and an infinite sequence of
  pairwise distinct markings $\vec{m}_0, \vec{m}_1, \ldots \in \N^P$
  with $h_{\Q}(\vec{m}_i) \leq b$ for every $i \geq 0$. Let $\vec{x}_i
  \in \Qnon^T$ be a solution to the state equation over $\Qnon$ that
  yields $h_{\Q}(\vec{m}_i)$, \ie\ such that $h_{\Q}(\vec{m}_i) =
  \sum_{t \in T} \lambda(t) \cdot \vec{x}_i(t)$ is minimized subject
  to
  \begin{align}
    \vec{m}_\text{target} &= \vec{m}_i + \sum_{t \in T} \vec{x}_i(t)
    \cdot \effect{t}.\label{eq:bounded:mk:eq}
  \end{align}
  
  Since $\N^P$ is well-quasi-ordered, there exist indices $i_0 < i_1 <
  \cdots$ such that $\vec{m}_{i_0} \leq \vec{m}_{i_1} \leq
  \cdots$. Since these markings are pairwise distinct, we may assume
  w.l.o.g.\ the existence of a place $p \in P$ such that
  $\vec{m}_{i_0}(p) < \vec{m}_{i_1}(p) < \cdots$ (otherwise, we could
  extract such a subsequence).

  Let us define the following constants:
  \[
    c \defeq \min\left\{\lambda(t) : t \in T\right\} \text{ and }
    d \defeq \frac{b \cdot |T| \cdot
             \max\left\{|\effect{t}(p)| : t \in T\right\}}{c}.
  \]
  Let $j \geq 0$ be such that $\vec{m}_\text{target}(p) -
  \vec{m}_{i_j}(p) < -d$. Such an index $j$ exists as $p$ takes
  arbitrarily large values along our infinite
  sequence. By~\eqref{eq:bounded:mk:eq}, we have:
  \begin{align*}
    \sum_{t \in T} \vec{x}_{i_j}(t) \cdot \effect{t}(p)
    &= \vec{m}_\text{target}(p) - \vec{m}_{i_j}(p)
    < -d.
  \end{align*}
  Thus, there exists $s \in T$ such that $\effect{s}(p) <
  0$ and $\vec{x}_{i_j}(s) > b / c$. Indeed, if it was not the case,
  it would be impossible to obtain a negative value smaller than $-d$.

  We are done since we obtain the following contradiction:
  \begin{align*}
    h_{\Q}(\vec{m}_{i_j})
    &= \sum_{t \in T} \lambda(t) \cdot \vec{x}_{i_j}(t)
    && \text{(by definition)} \\
    &\geq \lambda(s) \cdot \vec{x}_{i_j}(s)
    && \text{(by $\lambda(t) > 0$ and $\vec{x}_{i_j}(t) \geq 0$ for each
      $t \in T$)} \\
    &> \lambda(s) \cdot (b / c)
    && \text{(by $\lambda(s) > 0$ and $\vec{x}_{i_j}(s) > b / c$)} \\
    &\geq \lambda(s) \cdot (b / \lambda(s))
    && \text{(by $\lambda(s) \geq c$)} \\
    &= b \\
    &\geq h_{\Q}(\vec{m}_{i_j})
    && \text{(by boundedness).}\tag*{\qed}
  \end{align*}
\end{proof}

\section{Experimental results}

\Cref{fig:reach:all:prepruned} depicts an evaluation on reachability
instances where all tools were given the pruned Petri nets
(preprocessing time not included for any tool). The results are
essentially the same as those of~\Cref{fig:reach:all}.

\begin{figure}[h]
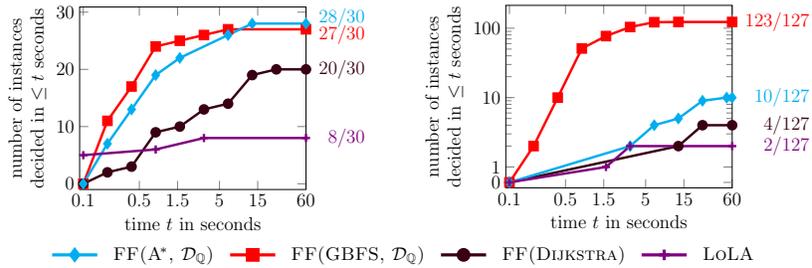

  \centering  
  \begin{minipage}[t]{0.45\textwidth}%
      \vspace*{0pt}%
      \input{benchmark-figs/fig-sypet-prepruned.tex}%
    \end{minipage}%
  \hspace*{5pt}%
  \begin{minipage}[t]{0.45\textwidth}%
      \vspace*{0pt}%
      \input{benchmark-figs/fig-reachability-prepruned.tex}%
  \end{minipage}%
  
  \begin{tikzpicture}[very thick, scale=\plotscale,
      every node/.style={scale=\plotscale}]
    \begin{customlegend}[legend columns=4,
        legend style={draw=none, column sep=1ex},
        legend entries={
          \textsc{FF(\astar, $\mathcal{D}_{\Q}$)},
          \textsc{FF(GBFS, $\mathcal{D}_{\Q}$)},
          \textsc{FF(Dijkstra)},
          \textsc{LoLA}},]
      \addlegend{FastForward QMarkingEQGurobi+Competitive}{diamond*}
      \addlegend{GBFS}{square*}
      \addlegend{FastForward zero+Competitive}{*}
      \addlegend{Lola}{+}
    \end{customlegend}
  \end{tikzpicture}
  \caption{Cumulative number of reachability instances decided over
    time (on pre-pruned instances). \emph{Left}: \textsc{sypet} suite
    (semi-log scale). \emph{Right}: \textsc{random-walk} suite (log
    scale).}\label{fig:reach:all:prepruned}
\end{figure}

\section{Structural distance}

\newcommand{\tin}[1]{\mathrm{in}(#1)}
\newcommand{\tout}[1]{\mathrm{out}(#1)}
\newcommand{\places}[1]{\supp{#1}}

All three $\G$-distances presented in the main text have an algebraic
flavor. While their complexity is significantly lower than the
non-elementary time complexity of Petri net reachability, they involve
solving optimization problems. An alternative avenue, mentioned in the
conclusion, consists in constructing less precise but more efficient
distance under-approximation based on structural properties.

We describe such a distance under-approximation adapted from the
syntactic distance of~\cite{S14} and related to the
``$\alpha$-graphs'' used by~\cite{FMWDR17}. Let $\pn = (P, T, f,
\lambda)$ be a weighted Petri net. The \emph{structural abstraction}
of $\pn$ is a weighted graph $G_\text{struct}(\pn)$ with places as
nodes with an edge $(p, t, q)$ iff transition $t$ consumes tokens from
$p$ and produces tokens into $q$. Since some transitions may consume
or produce no token, we imagine these as consuming from, or producing
to, an artificial ``sink place'' $\bot$. Intuitively, if $\vec{m}$ can
reach $\vec{m}'$, then each token of $\vec{m}$ must either make its
way to $\vec{m}'$ or disappear. Of course, tokens cannot move
independently and freely in $\pn$. However, paths in
$G_\text{struct}(\pn)$ yield a lower bound on an actual path from
$\vec{m}$ to $\vec{m}'$. A structural abstraction is given in
\Cref{fig:structural}.

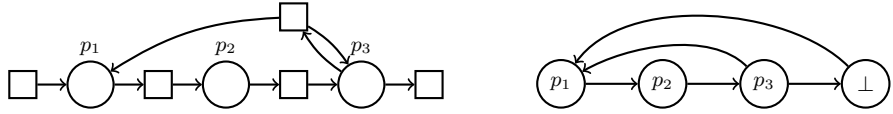
\begin{figure}[!h]
  \centering
  \begin{tikzpicture}[auto, thick, node distance=1cm, scale=0.9, transform shape]
  \tikzstyle{vertex} = [circle, draw, minimum width=20pt];
    
  \node[place,                   label={[yshift=-1pt]above:{$p_1$}}] (p1) {};
  \node[transition, right of=p1]                                     (t1) {};
  \node[place,      right of=t1, label={[yshift=-1pt]above:{$p_2$}}] (p2) {};
  \node[transition, right of=p2]                                     (t2) {};
  \node[place,      right of=t2, label={[yshift=-1pt]above:{$p_3$}}] (p3) {};
  \node[transition, right of=p3]                                     (t3) {};
  \node[transition, above of=t2]                                     (t4) {};
  \node[transition, left  of=p1]                                     (t5) {};
  
  \path[->]
  (t5) edge[] node[] {} (p1)
  (p1) edge[] node[] {} (t1)
  (t1) edge[] node[] {} (p2)
  (p2) edge[] node[] {} (t2)
  (t2) edge[] node[] {} (p3)
  (p3) edge[] node[] {} (t3)
  
  (p3) edge[bend  left=12] node[] {} (t4)
  (t4) edge[bend  left=12] node[] {} (p3)
  (t4) edge[bend right=15] node[] {} (p1)
  ;
\end{tikzpicture}%
\hspace*{35pt}%
\begin{tikzpicture}[auto, thick, node distance=1.5cm, scale=0.9, transform shape]
  \tikzstyle{vertex} = [circle, draw, minimum width=20pt];
 
  \node[vertex]              (p1)   {$p_1$};
  \node[vertex, right of=p1] (p2)   {$p_2$};
  \node[vertex, right of=p2] (p3)   {$p_3$};
  \node[vertex, right of=p3] (sink) {$\bot$};
 
  \path[->]
  (p1)   edge[] node[] {} (p2)
  (p2)   edge[] node[] {} (p3)
  (p3)   edge[] node[] {} (sink)

  (p3)   edge[out=135, in=30, looseness=0.8] node[] {} (p1)
  (sink) edge[out=135, in=55, looseness=0.8] node[] {} (p1)
  ;
\end{tikzpicture}
%
%
  \caption{\emph{Left:} A Petri net $\pn$. \emph{Right:} Its structural
    abstraction $G_\text{struct}(\pn)$.}\label{fig:structural}
\end{figure}

Formally, let $\tin{t} \defeq \{p \in P : f(p, t) > 0\}$ be the set of
input places of $t$ if it is nonempty, and $\tin{t} \defeq \{\bot\}$
otherwise; and let $\tout{t} \defeq \{p \in P : f(t, p) > 0\}$ be the
set of output places of $t$ if it is nonempty, and $\tout{t} \defeq
\{\bot\}$ otherwise. We define $G_\text{struct}(\pn) \defeq (V, E, T,
\mu)$ with $V \defeq P \cup \{\bot\}$, $\mu(p, t, q) \defeq
\lambda(t)$ and
\begin{align*}
  E &\defeq \left\{(p, t, q) : p \neq q, t \in T, p \in \tin{t} \text{
    and } \tout{t} \ni q\right\}.
\end{align*}

We obtain the \emph{structural distance} $d_\text{struct} \colon \N^P
\times \N^P \to \Qnon \cup \{\infty\}$ defined as follows. For every
marking $\vec{m}$, let $\places{\vec{m}} \defeq \{p \in P : \vec{m}(p)
> 0\} \cup \{\bot\}$ be the places marked in $\vec{m}$ together with
$\bot$ (considered permanently marked). Let:
\begin{align*}
  d_\text{struct}(\vec{m}, \vec{m}') &\defeq
  \max\left\{\kappa_{\vec{m}'}(p) : p \in \places{\vec{m}}\right\}, \text{
    where } \\
  \kappa_{\vec{m}'}(p) &\defeq \min\left\{\dist_{G_\text{struct}}(p, q)
  : q \in \places{\vec{m}'}\right\}.
\end{align*}

Informally, $\kappa_{\vec{m}'}(p)$ is the distance required to freely
move a token from place $p$ to a place marked in $\vec{m}'$, or to
destroy it. Since every token of $\vec{m}$ must achieve this task,
$d_\text{struct}$ maximizes $\kappa_{\vec{m}'}(p)$ among all places
marked in $\vec{m}$. Consider the Petri net of \Cref{fig:structural}
with $\vec{m} \defeq [p_1 \colon 0, p_2 \colon 1, p_3 \colon 1]$ and
$\vec{m}' \defeq [p_1 \colon 1, p_2 \colon 0, p_3 \colon 0]$. We have
$d_\text{struct}(\vec{m}, \vec{m}') = 2$ since $\kappa_{\vec{m}'}(p_2)
= 2$ and $\kappa_{\vec{m}'}(p_3) = 1$.

We show that $d_\text{struct}$ is an under-approximation by first
proving a lemma:

\begin{lemma}\label{lem:abstr}
  If $\vec{m} \reach{\sigma} \vec{m}'$, then for every $p \in
  \places{\vec{m}}$ there exists a path of weight at most
  $\lambda(\sigma)$ from $p$ to some $q \in \places{\vec{m}'}$ in
  $G_\text{struct}(\pn)$.
\end{lemma}

\begin{proof}
  We proceed by induction on $|\sigma|$. If $|\sigma| = 0$, then the
  claim follows immediately with the empty path. Assume $\sigma = t
  \tau$ with $t \in T$ and $\tau \in T^*$. There is some marking
  $\vec{m}''$ such that $\vec{m} \reach{t} \vec{m}'' \reach{\tau}
  \vec{m}'$. By induction hypothesis, for every $r \in
  \places{\vec{m}''}$, there exists a path $\path_r$ of weight at most
  $\lambda(\tau)$ from $r$ to some $q \in \places{\vec{m}'}$ in
  $G_\text{struct}(\pn)$. Let $p \in \places{\vec{m}}$. We must
  exhibit a path from $p$.

  If $p \in \places{\vec{m}''}$, then we are done as path $\path_p$
  satisfies $\mu(\path_p) \leq \lambda(\tau) \leq
  \lambda(\sigma)$. So, assume $p \not\in \places{\vec{m}''}$. By
  definition of $E$, we have $e \defeq (p, t, r) \in E$ for some $r
  \in \supp{\vec{m}''}$. Thus, path $\path \defeq e \path_r$ satisfies
  the claim since $\lambda(\path) = \lambda(t) + \mu(\path_r) \leq
  \lambda(t) + \lambda(\tau) = \lambda(\sigma)$.\qed
\end{proof}

\begin{proposition}\label{prop:struct:admissible}
  It is the case that $d_\text{struct}$ is a distance under-approximation.
\end{proposition}

\begin{proof}
  Let $\vec{m}, \vec{m}', \vec{m}'' \in \N^P$ be markings. We prove
  admissibility by establishing each property.

  \medskip\noindent\emph{Distance under-approximation.} We must show
  that $d_\text{struct}(\vec{m}, \vec{m}') \leq \dist_\pn(\vec{m},
  \vec{m}')$. Assume the latter differs from $\infty$, as we are
  otherwise done. Let $\sigma \in T^*$ be a shortest firing sequence
  such that
  \[
  \vec{m} \reach{\sigma} \vec{m}'.
  \]
  Let $p \in \places{\vec{m}}$ maximize $\kappa_{\vec{m}'}(p)$. By
  \Cref{lem:abstr}, $G_\text{struct}(\pn)$ has a path $\path$ of
  weight at most $\lambda(\sigma)$ from $p$ to some $q \in
  \places{\vec{m}'}$. Thus, $d_\text{struct}(\vec{m}, \vec{m}') =
  \kappa_{\vec{m}'}(p) \leq \dist_{G_\text{struct}(\pn)}(p, q) \leq
  \mu(\path) \leq \lambda(\sigma) = \dist_\pn(\vec{m}, \vec{m}')$.

  \medskip\noindent\emph{Triangle inequality.} We show
  $d_\text{struct}(\vec{m}, \vec{m}'') \leq d_\text{struct}(\vec{m},
  \vec{m}') + d_\text{struct}(\vec{m}', \vec{m}'')$. Assume that the
  right-hand side does not equal $\infty$ as we are otherwise
  done. Let $p, p' \in \places{\vec{m}}$ and $q \in \places{\vec{m}'}$
  respectively maximize $\kappa_{\vec{m}'}(p)$,
  $\kappa_{\vec{m}''}(p')$ and $\kappa_{\vec{m}''}(q)$.

  Let $q' \in \places{\vec{m}'}$ and $r \in \places{\vec{m}''}$ be
  such that $\kappa_{\vec{m}'}(p') = \dist_{G_\text{struct}}(p', q')$
  and $\kappa_{\vec{m}''}(q') = \dist_{G_\text{struct}}(q', r)$. Note
  that they are well-defined by $\kappa_{\vec{m}'}(p) \neq \infty$ and
  $\kappa_{\vec{m}''}(q) \neq \infty$.

  We have:
  \begin{alignat*}{3}
    &&& d_\text{struct}(\vec{m}, \vec{m}'') \\
    &&=\ & \kappa_{\vec{m}''}(p')\
    && \text{(by def.\ of $d_\text{struct}$)} \\
    &&\leq\ & \dist_{G_\text{struct}}(p', r)
    && \text{(by $r \in \places{\vec{m}''}$ and min.\ of
      $\kappa_{\vec{m}''}(p')$)} \\
    &&\leq\ & \dist_{G_\text{struct}}(p', q') + \dist_{G_\text{struct}}(q', r)\
    && \text{(by the triangle inequality)} \\
    &&=\ & \kappa_{\vec{m}'}(p') + \kappa_{\vec{m}''}(q') \\
    &&\leq\ & \kappa_{\vec{m}'}(p') + \kappa_{\vec{m}''}(q)
    && \text{(by $q' \in \places{\vec{m}'}$ and max.\ of $q$)} \\
    &&\leq\ & \kappa_{\vec{m}'}(p) + \kappa_{\vec{m}''}(q)
    && \text{(by $p' \in \places{\vec{m}}$ and max.\ of $p$)} \\
    &&=\ & d_\text{struct}(\vec{m}, \vec{m}') + d_\text{struct}(\vec{m}',
    \vec{m}'')
    && \text{(by def.\ of $d_\text{struct}$)}.
  \end{alignat*}

  \medskip\noindent\emph{Effectiveness.} The structural abstraction
  $G_\text{struct}(\pn)$ can be precomputed in linear time from $\pn$,
  and $\dist_{G_\text{struct}(\pn)}(p, q)$ can then be precomputed in
  polynomial time using \eg\ Dijkstra's algorithm.
  After these steps, $d_\text{struct}(\vec{m}, \vec{m}')$ can
  be evaluated in time $\mathcal{O}(|\places{\vec{m}}| \cdot
  |\places{\vec{m}'}|)$.\qed
\end{proof}

Let us stress that $d_\text{struct}(\vec{m}, \vec{m}')$ yields a crude
estimation of $\dist_\pn(\vec{m}, \vec{m}')$. Indeed, its value is
always upper bounded by $|P| \cdot \max\{\lambda(t) : t \in T\}$,
while the actual distance could be arbitrarily large in $\vec{m}$ and
$\vec{m}'$. Nevertheless, it is lightweight since it enables
pre-computations. This makes it useful in particular for reachability
graphs with short paths but large branching factors.

For example, instances from the \textsc{sypet} suite have a large
branching factor. They have between 23 and 187 unguarded transitions.
Most markings tend to enable some guarded transitions as well, so the
average branching factor is larger. In particular, the branching
factor of initial markings ranges from 30 to 300.\footnote{Chess and
  Go respectively have an average branching factor of $\sim$35 and
  $\sim$350~\cite{RN09}.}

\begin{figure}[h]
\hspace*{\fill}%
\begin{minipage}[t]{0.45\textwidth}
  \begin{center}
    \vspace{0pt}
    \plotcanvaslonger{%
      time $t$ in seconds}{%
      number of instances decided \\ in $\leq t$ seconds}{%

\toolplotwithtotalcount{A* in competitive mode with heuristic QMarkingEQGurobi}{diamond*}{(100, 0)(200, 4)(400, 12)(800, 18)(1600, 21)(3200, 22)(6400, 25)(12800, 28)(102400, 29)(204800, 30)(600000, 30)}{600000}{30}{\textcolor{colA* in competitive mode with heuristic QMarkingEQGurobi}{30/30}};

\toolplotwithtotalcount{best-first in competitive mode with heuristic syntactic}{pentagon*}{(100, 0)(200, 10)(400, 17)(800, 23)(1600, 25)(6400, 27)(102400, 28)(600000, 28)}{600000}{28}{\textcolor{colbest-first in competitive mode with heuristic syntactic}{28/30}};

\toolplotwithtotalcount{Lola}{+}{(100, 0)(200, 1)(400, 2)(800, 5)(12800, 6)(25600, 8)(204800, 9)(409600, 16)(600000, 16)}{600000}{16}{\textcolor{colLola}{16/30}};

\toolplotwithtotalcount{KReach}{o}{(100, 0)(600000, 0)}{600000}{0}{\textcolor{colKReach}{0/30}};

    }
    \begin{tikzpicture}[very thick, scale=\plotscale,
                        every node/.style={scale=\plotscale}]
      \begin{customlegend}[legend columns=2,
          legend style={draw=none, column sep=1ex, cells={align=center}, anchor=center},
          legend entries={\textsc{FF(\astar, $\mathcal{D}_{\Q}$)}, \textsc{FF(GBFS, $\mathcal{D}_\text{struct}$)}, \textsc{LoLA}, \textsc{KReach}}]
\addlegend{A* in competitive mode with heuristic QMarkingEQGurobi}{diamond*}
\addlegend{best-first in competitive mode with heuristic syntactic}{pentagon*}
\addlegend{Lola}{+}
        \addlegend{KReach}{o}
      \end{customlegend}
    \end{tikzpicture}%
  \end{center}
\end{minipage}%
\hfill
\begin{minipage}[t]{0.45\textwidth}
  \centering
  \vspace{0pt}
  \begin{tikzpicture}[scale=\plotscale, every node/.style={scale=0.9}]%
      \begin{loglogaxis}[%
          xmin=50,
          xmax=1000000,
          ymin=50,
          ymax=1000000,
          width=\textwidth,
          xlabel style={align=center},
          xlabel={time in seconds \\ for \textsc{FF(\astar, $\mathcal{D}_{\Q}$)}},
          xticklabels={0, 0.1, 1, 10, 100, 1000, 10000},
          yticklabels={0, 0.1, 1, 10, 100, 1000, 10000},
          ylabel style={align=center},
          ylabel=time in seconds \\  for \textsc{FF(GBFS, $\mathcal{D}_\text{struct}$)},
          legend style={
            at={(0, 0.65)},
            anchor=west},
          reverse legend,
        ]%
        \addplot[mark=none] coordinates {(100,100) (600000, 600000)};
          \addplot[gray, mark=none] coordinates {(600000,100) (600000, 600000)};
          \addplot[gray, mark=none] coordinates {(100,600000) (600000, 600000)};
        
\addplot[only marks, red] coordinates {
};

\addplot[only marks, blue, mark={+}] coordinates {(7104.154348373413, 722.2304344177246) 
(1038.7821197509766, 415.5278205871582) 
(213.79947662353516, 160.92300415039062) 
(172.26862907409668, 142.68732070922852) 
(638.8018131256104, 322.39556312561035) 
(204.7886848449707, 147.34315872192383) 
(212.8469944000244, 206.15601539611816) 
(169.34537887573242, 142.81010627746582) 
(481.6296100616455, 269.5479393005371) 
(709.1202735900879, 5322.7269649505615) 
(69038.33341598511, 600000) 
(105580.61003684998, 600000) 
(173.36416244506836, 145.31707763671875) 
(645.2891826629639, 5260.328769683838) 
(266.28661155700684, 178.07960510253906) 
(749.7358322143555, 287.6255512237549) 
(5251.090049743652, 816.3259029388428) 
(226.36795043945312, 166.17274284362793) 
(10305.615425109863, 76216.77875518799) 
(5274.735927581787, 820.3275203704834) 
(4591.395139694214, 607.0621013641357) 
(216.02535247802734, 137.85099983215332) 
(433.22110176086426, 255.39278984069824) 
(8486.419677734375, 665.6458377838135) 
(254.04644012451172, 160.8717441558838) 
(880.650520324707, 438.4024143218994) 
(375.60081481933594, 227.68211364746094) 
(199.052095413208, 179.46529388427734) 
(1376.4584064483643, 263.37194442749023) 
(2339.8454189300537, 488.19947242736816) 
};

      \end{loglogaxis}%
    \end{tikzpicture}%
\end{minipage}
\caption{Results on the \textsc{sypet} suite with a time limit of 600
  seconds. \emph{Left:} Cumulative number of instances shown
  reachable. \emph{Right:} Performance comparison per instance of
  \textsc{FastFoward} with two different schemes. Marks on the
  \textcolor{gray!80!black}{gray} lines denote
  timeouts.}\label{fig:sypet:appendix}
\end{figure}
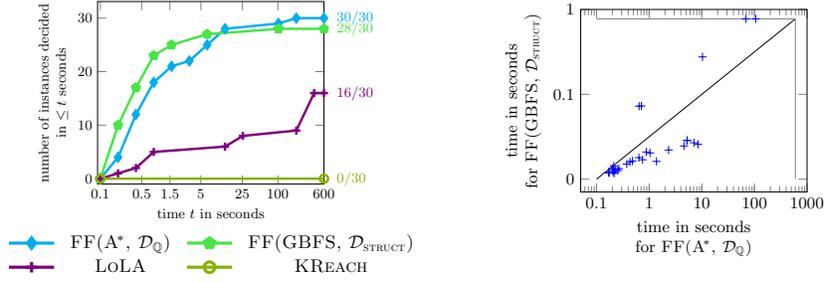

Let $\mathcal{D}_\text{struct}$ be distance under-approximation scheme
obtained from the structural distance. This scheme is not unbounded,
but can still be used with GBFS without termination
guarantee. \Cref*{fig:sypet:appendix} compares the performance of
\textsc{FastForward} using \astar\ with $\mathcal{D}_\Q$ and using
GBFS with $\mathcal{D}_\text{struct}$ on a time limit of 600
seconds. The former is faster on most instances, but it is vastly
outperformed by \astar\ on a few instances. An explanation is provided
by the large branching factor and short paths, and how these emphasize
the characteristics of the different approaches. Note that the
structural abstraction can be precomputed. On the other hand,
\astar\ requires computing the heuristic on each successor before the
next node is chosen for expansion. It thus is at a slight disadvantage
on instances where a shortest witness is so short that it is found
rather quickly even with the coarse structural distance. Its advantage
is on the instances where the length of a shortest witness is at the
upper end of the range. There, the large branching factor fully comes
into play and a search algorithm must more aggressively discard parts
of the search space.

\end{document}